\documentclass[12pt]{article}

\usepackage[utf8]{inputenc}
\usepackage{tabularray}
\usepackage[export]{adjustbox}

\usepackage[sectionbib]{natbib}
\usepackage{array,epsfig,rotating}
\usepackage{hyperref}
\hypersetup{
  colorlinks,
  linkcolor={red!50!black},
  citecolor={blue},
  urlcolor={blue!80!black}
}
\usepackage{color}
\usepackage[dvipsnames]{xcolor}
\usepackage{tabularx}
\usepackage{amsmath}
\usepackage{amssymb}
\usepackage{amsfonts}
\usepackage{amsthm}
\usepackage{float}
\usepackage{verbatim}
\usepackage{subcaption}
\usepackage[ruled, vlined, lined, commentsnumbered]{algorithm2e}
\usepackage[nodisplayskipstretch]{setspace}
\usepackage{booktabs}
\usepackage{graphicx}
\usepackage{makecell}
\usepackage{multirow}
\usepackage[noabbrev,capitalize]{cleveref}
\usepackage[shortlabels]{enumitem}
\setlist{itemsep=1pt, topsep=2pt, parsep=0pt, partopsep=0pt}

\usepackage[margin=1in]{geometry}

\usepackage{forest}
\newtheorem{theorem}{Theorem}
\newtheorem{assumption}{Assumption}
\newtheorem{lemma}{Lemma}

\newtheorem{proposition}{Proposition}

\newtheorem{remark}{Remark}

\newtheorem*{theorem*}{Theorem}

\newcommand{\indep}{\perp \!\!\! \perp}
\newcommand{\R}{\mathbb R}
\newcommand{\RR}{\mathbf R}
\newcommand{\XX}{\mathbf X}
\newcommand{\ZZ}{\mathbf Z}
\newcommand{\UU}{\mathbf U}
\newcommand{\WW}{\mathbf W}
\newcommand{\GG}{\mathbf G}
\newcommand{\II}{\mathbf I}
\newcommand{\PP}{\mathbb P}
\newcommand{\bpi}{\boldsymbol{\pi}}
\newcommand{\btheta}{\boldsymbol{\theta}}
\newcommand{\bTheta}{\boldsymbol{\Theta}}
\newcommand{\bDelta}{\boldsymbol{\Delta}}
\newcommand{\bSigma}{\boldsymbol{\Sigma}}
\newcommand{\bOmega}{\boldsymbol{\Omega}}
\newcommand{\bLambda}{\boldsymbol{\Lambda}}
\newcommand{\bepsilon}{\boldsymbol{\epsilon}}

\newcommand{\sumin}{\sum_{i=1}^N}

\newcommand{\EE}{\mathbb E}

\def\sgn{\textnormal{sgn}}
\def\hat{\widehat}
\def\tilde{\widetilde}

\newcommand{\cbr}[1]{\left\{ #1 \right\}}

\DeclareMathOperator*{\argmin}{argmin}
\DeclareMathOperator*{\argmax}{argmax}

\usepackage{tikz}
\usetikzlibrary{shapes, calc, shapes, arrows, positioning}
\tikzstyle{qedge}=[->,thick,black]
\tikzstyle{pre}=[->,thick,dotted]

\definecolor{myblue}{rgb}{0.0265,    0.6137,    0.8135}
\definecolor{myyellow}{rgb}{0.9290,    0.6940,    0.1250}
\tikzstyle{neuron}=[draw,circle,minimum size=26pt,inner sep=0pt, fill=black!10]
\tikzstyle{hidden}=[draw,circle,minimum size=26pt,inner sep=0pt, fill=white]
\usetikzlibrary{arrows.meta}
\tikzset{>={Latex[width=3mm,length=2mm]}}
\tikzstyle{arr}=[->, thick, black]
\usetikzlibrary{decorations.text}
\usetikzlibrary{shadings,shapes.symbols}
\tikzset{
    double color fill/.code 2 args={
        \pgfdeclareverticalshading[%
            tikz@axis@top,tikz@axis@middle,tikz@axis@bottom%
        ]{diagonalfill}{100bp}{%
            color(0bp)=(tikz@axis@bottom);
            color(50bp)=(tikz@axis@bottom);
            color(50bp)=(tikz@axis@middle);
            color(50bp)=(tikz@axis@top);
            color(100bp)=(tikz@axis@top)
        }
        \tikzset{shade, left color=#1, right color=#2, shading=diagonalfill}
    }
}

\newcommand\blfootnote[1]{%
  \begingroup
  \renewcommand\thefootnote{}\footnote{#1}%
  \addtocounter{footnote}{-1}%
  \endgroup
}

\title{Beyond Local Independence: High-Dimensional Latent Class Graphical Models with Shared Block Structure}
\author{Seunghyun Lee and Yuqi Gu}
\date{Department of Statistics, Columbia University}

\def\spacingset#1{\renewcommand{\baselinestretch}%
{#1}\small\normalsize} 
\allowdisplaybreaks

\begin{document}
\maketitle
\blfootnote{Emails: \texttt{sl4963@columbia.edu}, \texttt{yuqi.gu@columbia.edu}}
\vspace{-5mm}
\begin{abstract}
Latent class models are central tools for multivariate categorical data from heterogeneous populations, but their standard local-independence assumption is often unrealistic in modern high-dimensional applications. We propose a high-dimensional latent class graphical model for ordinal responses with block-structured local dependence. The model retains the interpretability and parsimony of classical latent class analysis by imposing a shared block partition of variables, while allowing class-specific graphical dependence within each block. We develop a scalable three-step estimator that first recovers latent classes by spectral clustering of a flattened response matrix, then estimates class-specific latent covariance matrices and aggregates them to recover the shared block partition, and finally estimates sparse within-block precision matrices. We establish finite-sample error bounds for clustering, covariance estimation, block recovery, and precision-matrix estimation, yielding end-to-end consistency of all model components under high-dimensional scaling. Simulations demonstrate accurate recovery of latent classes, the shared block partition, and class-specific dependence graphs with scalable computation. Applications to American National Election Studies survey data and HapMap3 genotype data show that the method uncovers interpretable local dependence structures while accounting for latent heterogeneity.
\end{abstract}


\spacingset{1.6}

\section{Introduction}

Modern social-science and biomedical studies often collect high-dimensional ordinal or categorical measurements in the presence of unobserved population heterogeneity.
Latent class models (LCMs) are widely used mixture models for such data: each individual belongs to a latent class, and the class-specific multivariate response distributions capture heterogeneity across latent subpopulations. Since their introduction in \cite{lazarsfeld1950logical}, LCMs have become standard tools for clustering and modeling multivariate categorical responses in social science, education, and related fields \citep{magidson2020latent,vermunt2002latent,wang2011latent,korpershoek2015differences,lyu2025degree,Lyu2026latent}.

A central simplifying assumption in traditional LCMs is \emph{local independence}: conditional on the latent class, the observed multivariate responses are mutually independent. This assumption makes the model parsimonious and computationally tractable, but it is often too restrictive in applications, especially for high-dimensional responses. For example, in public opinion surveys, items concerning the same political topic, using similar wording, or appearing in related questionnaire modules may remain dependent even after conditioning on a respondent's latent political subgroup.
In genetic SNP genotype data, nearby genetic markers often exhibit dependence due to \emph{linkage disequilibrium}, producing block-like correlation patterns along the genome \citep{barrett2005haploview,gabriel2002structure}. In both examples, local dependence arises from common, interpretable groupings of variables. Ignoring such dependence can distort parameter estimates and latent class assignments, as well as limit scientific interpretation %
\citep{vacek1985effect,hagenaars1988latent}.

This challenge of relaxing local independence is further amplified in modern high-dimensional data with an \emph{unknown dependence} structure. As the number of responses $J$ grows, estimating the dependence structure among the $J$ variables becomes computationally challenging. 
Existing approaches to local dependence in LCMs provide important modeling tools, but they are typically developed for lower-dimensional settings, rely on prespecified or heavily parameterized (unstructured) dependence structures, and do not provide statistical guarantees \citep[][see Supplement \ref{supp:literature} for a detailed discussion]{hagenaars1988latent,qu1996random,lee2020detecting,lee2022estimating,mazaheri2023causal}. %
Conversely, high-dimensional Gaussian-copula graphical models are designed to learn dependence among variables \citep{liu2012high,danaher2014joint,guo2015graphical,fan2017high,feng2019high}, but generally do not address unknown heterogeneous latent classes. Thus, there remains a need for principled statistical methodology that jointly learns latent heterogeneity and structured local dependence from high-dimensional ordinal responses.
Motivated by this gap, we propose a high-dimensional latent class graphical model with shared block structure. Conditional on the latent class, ordinal responses are modeled through a latent Gaussian copula, whose covariance matrix is block diagonal after a shared permutation of variables. The shared block partition captures interpretable groups of locally dependent variables, while the within-block graphical structure is allowed to vary across latent classes.
Going back to the survey example, this block structure represents item groups corresponding to different survey topics. Within each block, the dependence structure is allowed to vary flexibly across classes, capturing class-specific heterogeneity. Thus, our modeling strategy provides a middle ground between the classical local independent LCM and an unrestricted mixture of graphical models, where the latter is difficult to estimate and interpret in high dimensions.

Methodologically, we introduce a computationally efficient three-step estimation pipeline to learn the shared block structure as well as all other model components (such as the latent classes and latent covariance/precision matrices) in an unsupervised manner. First, we use the low-rank structure of the response matrix and apply \emph{spectral clustering} of the flattened data matrix to learn the latent class labels. Second, within each estimated latent class, we estimate the latent covariance matrix using pairwise \emph{polychoric correlations}, and then aggregate the resulting sparsity patterns across classes to recover the shared block partition. Finally, after the blocks have been learned, we estimate the finer conditional dependence structures within them through sparse \emph{precision-matrix estimation}.
Each step is computationally feasible for large $J$, avoiding the full likelihood optimization that would be required by an unrestricted locally dependent LCM. 
Overall, this yields an unsupervised, scalable, and theoretically justified approach for learning latent classes, shared block structure, and class-specific graphical dependence in high dimensions.
Importantly, our theory provides end-to-end guarantees: we derive non-asymptotic error bounds for each estimation stage and establish consistency of the complete pipeline under a high-dimensional scaling.

Simulation studies show that the proposed method accurately recovers latent classes, shared block structures, and class-specific precision matrices across a range of high-dimensional settings. We demonstrate the proposed methodology in two applications that motivate the study: public-opinion survey responses from the American National Election Studies and SNP genotype data from HapMap3. In the survey application, the learned blocks correspond to interpretable item groups such as political engagement and racism. In the genotype application, the method detects dependent blocks among SNPs in the presence of unknown, heterogeneous subpopulations. Together, these examples illustrate that the proposed framework can uncover interpretable block dependence structures while accounting for latent heterogeneity.

The remainder of this paper is organized as follows. Section \ref{sec:setup} proposes the block-dependent latent class model after reviewing the standard latent class model. Section \ref{sec:theory and method} details the three-step estimation methodology along with theoretical guarantees and implementation details. Section \ref{sec:sims} conducts simulations to evaluate the statistical performance and computational efficiency of the proposed method. Section \ref{sec:data} demonstrates the real-world utility of our method by applying it to two diverse real datasets, one from social science surveys and one from genetics. Section \ref{sec:discussion} concludes the paper by discussing avenues for future research. The R code for the proposed method is available at \url{https://github.com/seunghyun-stats/Block-dependent-LCM}. 

\paragraph{Notation.}
For any positive integer $K$, set $[K] := \{1, \ldots, K\}$. Let $\Delta^K$ denote the $K$-dimensional probability simplex. Given a finite set $A$, let $\mathcal{S}_A$ denote the set of permutations on $A$. 
Let $\Phi/\phi$ denote the c.d.f./p.d.f. of the standard normal, respectively. Similarly, let $\Phi_2$ denote the c.d.f. of a bivariate normal $(X_1, X_2)^\top$ where $X_1, X_2$ have standard normal marginals and correlation $\rho$, i.e., $\Phi_2(u,v,\rho) := \PP(X_1 \le u, ~X_2 \le v).$
For a $J \times J$ matrix $\mathbf M$, denote the operator norm (spectral norm) as $\|\mathbf M\|$. Denote the entry-wise maximum norm as $\|\mathbf M\|_{\max} := \max_{j, j' \in [J]} |\mathbf M_{(j,j')}|$. For a (not necessarily square) matrix $\mathbf M$, denote its $k$th largest singular value as $\sigma_k(\mathbf M)$. 
Throughout the paper, $c_1, c_2, \ldots$ will denote constants whose exact values may vary across contexts.

\section{Model Setup}\label{sec:setup}
\subsection{Latent class models}\label{subsec:lcm}
We first introduce the standard LCM framework. For $N$ individuals and $J$ ordinal items, let $R_{i,j} \in [C_j]$ denote the response of individual $i$ to item $j$, where $C_j \ge 2$ is the number of categories for item $j$. Denote the aggregate $N \times J$ response matrix as $\RR$.

Each individual has a latent class $Z_i \in [K]$, generated with proportion parameter $\bpi = (\pi_1, \ldots, \pi_K)^\top \in \Delta^K$:
\begin{equation}\label{eq:latent class Z}
    Z_i \stackrel{i.i.d.}{\sim} \textsf{Categorical}(\bpi).
\end{equation}
In other words, we assume $\PP(Z_i = k) = \pi_k$.
Conditional on the latent class $Z_i$, the item responses $R_{i,j}$ are generated independently, i.e., $R_{i,1} \indep \ldots \indep R_{i,J} \mid Z_i$:
\begin{equation}\label{eq:LCM response}
    R_{i,j} \mid (Z_i = k) \stackrel{ind.}{\sim} \textsf{Categorical}(\btheta_{j,k}), \quad \forall k \in [K],
\end{equation}
where $\btheta_{j,k} = (\theta_{j,1,k},\ldots,\theta_{j,C_j,k})^\top \in \Delta^{C_j}$ is the class-specific probability mass function for item $j$. These vectors differ across $k$, reflecting heterogeneity across latent classes.

As discussed in the introduction, despite its modeling convenience, this \emph{local independence} assumption in \eqref{eq:LCM response} is often violated in practice. Next, we describe the proposed model that allows structured conditional dependence.

\subsection{Block-dependent latent class models}\label{subsec:block-dependent LCM}

We define the model through the generative process for $\RR$. The usual LCM is generalized by viewing ordinal responses as discretizations of a latent Gaussian; see \cref{fig:data-generating} for a graphical illustration. Given a latent class $Z_i = k$ as in \eqref{eq:latent class Z}, generate a latent Gaussian vector $\XX_i = (X_{i,1},\ldots,X_{i,J})^\top$:
\begin{equation}\label{eq:X given Z}
    \XX_i \mid Z_i = k \stackrel{i.i.d.}{\sim} \mathcal{N}_J(\mathbf 0, \bSigma_k).
\end{equation}
Here, without loss of generality, assume $\text{diag}(\bSigma_k) = \mathbf 1_J$ for all $k \in [K]$. The observed ordinal responses are obtained by thresholding $X_{i,j}$ at $\bDelta_{j,k} = (\Delta_{j,1,k}, \ldots, \Delta_{j,C_j-1,k})^\top$:
\begin{align}\label{eq:R}
    R_{i,j} = \begin{cases}
        1 & \text{if} ~~ X_{i,j} < \Delta_{j,1,k}, \\
        2 & \text{if} ~~ \Delta_{j,1,k} \le X_{i,j} < \Delta_{j,2,k}, \\
        \ldots,\\
        C_j & \text{if} ~~  \Delta_{j,C_j-1,k} \le X_{i,j}.
    \end{cases}
\end{align}
Here $\Delta_{j,1,k} < \cdots < \Delta_{j,C_j-1,k}$, so larger latent Gaussian values produce larger ordinal responses. Write $\bDelta_k := (\bDelta_{1,k}^\top, \ldots, \bDelta_{J,k}^\top)^\top$ and set $\Delta_{j,0,k} = -\infty$, $\Delta_{j,C_j,k} = \infty$. This latent-threshold formulation is standard for ordinal responses \citep{muthen1984general,uebersax1999probit}; special cases include \cite{guo2015graphical}, which corresponds to $K=1$, and \cite{lee2022estimating}, which assumes class-homogeneous thresholds $\bDelta_{j,k} \equiv \bDelta_j$. When clear from context, we sometimes drop the individual-level index $i$ and write $Z, \XX, \RR$.

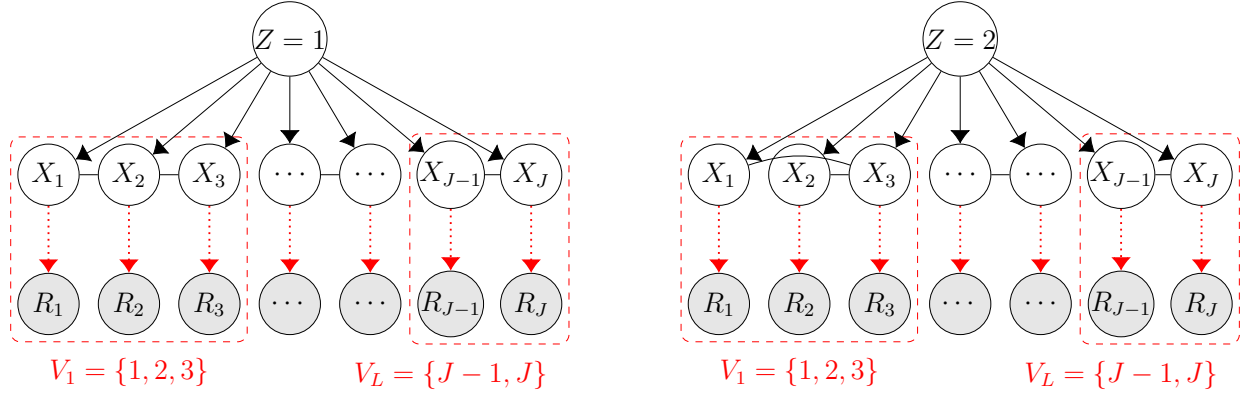
\begin{figure}[h!]
\centering
\resizebox{\textwidth}{!}{
  \begin{tikzpicture}[node distance={12mm},
                      empty/.style = {}
                     ] 

  \begin{scope}
  \node[neuron] (R11) {$R_1$}; 
  \node[neuron] (R12) [right of=R11]  {$R_2$}; 
  \node[neuron] (R13) [right of=R12]  {$R_3$}; 
  \node[neuron] (R14) [right of=R13]  {$\cdots$}; 
  \node[neuron] (R15) [right of=R14]  {${\cdots}$}; 
  \node[neuron] (R16) [right of=R15]  {$R_{J-1}$};
  \node[neuron] (R17) [right of=R16]  {$R_J$};

  \node[hidden] (X11) [above=1cm of R11] {$X_1$}; 
  \node[hidden] (X12) [right of=X11]  {$X_2$}; 
  \node[hidden] (X13) [right of=X12]  {$X_3$}; 
  \node[hidden] (X14) [right of=X13]  {${\cdots}$}; 
  \node[hidden] (X15) [right of=X14]  {${\cdots}$}; 
  \node[hidden] (X16) [right of=X15]  {$X_{J-1}$};
  \node[hidden] (X17) [right of=X16]  {$X_J$};
 
  \node[hidden] (Z1) [above=1cm of X14] {$Z = 1$};    

  \node[empty,red] (V11) [below=0.2cm of R12] {$V_1 = \{1,2,3\}$}; 
  \node[empty,red] (V12) [below=0.2cm of R16] {$V_L = \{J-1, J\}$}; 

  \node [draw=red,dashed,fit=(X11)(X12)(X13)(R11)(R12)(R13), inner sep=0.1cm, rounded corners] {};
  \node [draw=red,dashed,fit=(X16)(X17)(R16)(R17), inner sep=0.1cm, rounded corners] {};
  
  \draw[->] (Z1) -- (X11); \draw[->] (Z1) -- (X12); \draw[->] (Z1) -- (X13); \draw[->] (Z1) -- (X14); \draw[->] (Z1) -- (X15); \draw[->] (Z1) -- (X16); \draw[->] (Z1) -- (X17);

  \draw[->, red, pre] (X11) -- (R11); \draw[->, red, pre] (X12) -- (R12); \draw[->, red, pre] (X13) -- (R13); \draw[->, red, pre] (X14) -- (R14); \draw[->, red, pre] (X15) -- (R15); \draw[->, red, pre] (X16) -- (R16); \draw[->, red, pre] (X17) -- (R17);

  \draw[] (X11) -- (X12); \draw[] (X12) -- (X13); \draw[] (X14) -- (X15); \draw[] (X16) -- (X17);
  \end{scope}

  \begin{scope}[xshift=10cm]
  \node[neuron] (R21) {$R_1$}; 
  \node[neuron] (R22) [right of=R21]  {$R_2$}; 
  \node[neuron] (R23) [right of=R22]  {$R_3$}; 
  \node[neuron] (R24) [right of=R23]  {$\cdots$}; 
  \node[neuron] (R25) [right of=R24]  {${\cdots}$}; 
  \node[neuron] (R26) [right of=R25]  {$R_{J-1}$};
  \node[neuron] (R27) [right of=R26]  {$R_J$};

  \node[hidden] (X21) [above=1cm of R21] {$X_1$}; 
  \node[hidden] (X22) [right of=X21]  {$X_2$}; 
  \node[hidden] (X23) [right of=X22]  {$X_3$}; 
  \node[hidden] (X24) [right of=X23]  {${\cdots}$}; 
  \node[hidden] (X25) [right of=X24]  {${\cdots}$}; 
  \node[hidden] (X26) [right of=X25]  {$X_{J-1}$};
  \node[hidden] (X27) [right of=X26]  {$X_J$};
 
  \node[hidden] (Z2) [above=1cm of X24] {$Z = 2$};    

  \node[empty,red] (V21) [below=0.2cm of R22] {$V_1 = \{1,2,3\}$}; 
  \node[empty,red] (V22) [below=0.2cm of R26] {$V_L = \{J-1, J\}$}; 

  \node [draw=red,dashed,fit=(X21)(X22)(X23)(R21)(R22)(R23), inner sep=0.1cm, rounded corners] {};
  \node [draw=red,dashed,fit=(X26)(X27)(R26)(R27), inner sep=0.1cm, rounded corners] {};
  
  \draw[->] (Z2) -- (X21); \draw[->] (Z2) -- (X22); \draw[->] (Z2) -- (X23); \draw[->] (Z2) -- (X24); \draw[->] (Z2) -- (X25); \draw[->] (Z2) -- (X26); \draw[->] (Z2) -- (X27);

  \draw[->, red, pre] (X21) -- (R21); \draw[->, red, pre] (X22) -- (R22); \draw[->, red, pre] (X23) -- (R23); \draw[->, red, pre] (X24) -- (R24); \draw[->, red, pre] (X25) -- (R25); \draw[->, red, pre] (X26) -- (R26); \draw[->, red, pre] (X27) -- (R27);

  \draw[] (X21) to[bend left=18] (X23); \draw[] (X22) -- (X23); \draw[] (X24) -- (X25); \draw[] (X26) -- (X27);
  \end{scope}

  \end{tikzpicture}
}
  \caption{Graphical model illustration of the data generating process. Here, $Z$ denotes the categorical latent class, $\XX = (X_1, \ldots, X_{J})^\top$ denotes the Gaussian latent vector, and $\RR = (R_1, \ldots, R_{J})^\top$ denotes the ordinal response vector. The dotted red boxes indicate the shared block structure, while the within-block graphical structure may vary across latent classes. We omit the individual-level index $i$ for simplicity.} 
  \label{fig:data-generating}
\end{figure}

Our key modeling assumption is a \emph{shared block dependence} structure across all $K$ mixture components.  To elaborate, assume that the $J \times J$ covariance matrices $\bSigma_k = \big(\sigma_{(j,j'),k}\big)_{j, j' \in [J]}$ are block-diagonal after appropriately permuting the row/column indices:
\begin{align}\label{eq:sigma_k block structure}
\bSigma_k = \begin{pmatrix}
   \bSigma_{k}^{(1)} & 0 & \cdots & 0 \\
    0 & \bSigma_{k}^{(2)} & \cdots & 0 \\
    \vdots & \vdots & \ddots & \vdots \\
    0 & 0 & \cdots & \bSigma_{k}^{(L)}
\end{pmatrix} = \text{diag}(\bSigma_k^{(1)}, \ldots, \bSigma_k^{(L)}),\quad k=1,\ldots,K.
\end{align}
Here, each $\bSigma_k^{(\ell)}$ is positive definite with size $J_{\ell} \times J_{\ell}$, $L$ is the number of blocks, and $\sum_{\ell=1}^L J_\ell = J$. For notational convenience, define a vertex set $V = [J] = \cup_{\ell=1}^L V_\ell$, so that each $V_\ell$ denotes the vertices belonging to the $\ell$th block and $|V_\ell| = J_\ell$.

To ensure that the block structure is well-defined (in the sense that there is no finer possible partition of $[J]$), we impose the following assumption. 

\begin{assumption}\label{assump:minimal strength}
For any $\ell \in [L]$, define a graph on $V_\ell$ with edges $E_\ell = \{(j,j'): j \neq j' \in V_{\ell}, ~ \max_k |\sigma_{(j,j'),k}| > \epsilon_N\}$, where $\epsilon_N >0$ is a quantity that is allowed to converge to $0$ as $N$ grows. Assume that the graph $(V_\ell, E_\ell)$ is connected. 
\end{assumption} 
\noindent The connectivity in \cref{assump:minimal strength} %
ensures local dependence within each block. In contrast, items in different blocks are constrained to be conditionally independent.

To further understand finer conditional dependence relationships within each block, it is helpful to introduce the precision matrices $\bOmega_k = \bSigma_k^{-1} = \big(\omega_{(j,j'),k}\big)_{j, j' \in [J]}$:
$$\bOmega_k = \text{diag}(\bOmega_k^{(1)}, \ldots, \bOmega_k^{(L)})
,\quad k=1,\ldots,K.$$ 
Here, the block diagonal structure of $\bOmega_k$ follows immediately from that of $\bSigma_k$ in \eqref{eq:sigma_k block structure}. Since $\XX \mid (Z=k)$ follows a Gaussian graphical model with precision matrix $\bOmega_k$, the usual conditional independence relationships follow: for $j \neq j'$, $\omega_{(j,j'),k} = 0$ if and only if $X_j \indep X_{j'} \mid (X_{[J] \setminus \{j,j'\}}, Z=k)$. For notational convenience, denote the collections of all covariance and precision matrices as $\bSigma = (\bSigma_{1}, \ldots, \bSigma_{K})$ and $\bOmega = (\bOmega_{1}, \ldots, \bOmega_{K})$, respectively.

The classical locally independent LCM (see \cref{subsec:lcm}) is recovered by taking $L=J$ and $J_\ell = 1$ for each $\ell\in[L]$. The covariance matrices simply reduce to $\bSigma_k^{(\ell)} = 1$, and we recover the conditional independence of the responses. This is equivalent to directly parametrizing the conditional response probabilities, e.g. $\PP(R_j=a\mid Z=k)=\theta_{j,a,k}$ as in \eqref{eq:LCM response}.

\subsection{Assumptions}
We describe two general assumptions that will be imposed for all theoretical results. We first assume that the number of ordinal categories is bounded. This is a natural assumption in the context of our applications, such as surveys or ratings, for which the number of categories is fundamentally bounded. 
\begin{assumption}[Bounded number of categories]\label{assmp:bounded number of categories}
    Assume that $C_{\max} := \max_{j \in [J]} C_j < \infty$.
\end{assumption}

Next, we impose some assumptions regarding the parameter values.

\begin{assumption}[Bounded parameter space]\label{assmp:bounded parameter}
    Assume that there exist constants $\delta_1, \delta_2, \delta_3, M > 0$ such that, for every $k \in [K]$ and every $j \neq j' \in [J]$,
    $$|\sigma_{(j,j'),k}| \le 1-\delta_1, \quad \max_{a \in [C_j-1]} |\Delta_{j,a,k}| \le M, \quad \min_{2 \le a \le C_j -1} |\Delta_{j,a,k} - \Delta_{j,a-1,k}| \ge \delta_2, \quad \pi_k \ge \frac{\delta_3}{K}.$$ %
\end{assumption}

\noindent Here, the first part regarding $\bSigma_k$ rules out the degeneracy where two entries of the latent Gaussian $\XX$ are nearly collinear. %
The second and third parts require minimum separation between the thresholds $\Delta_{j,a}$, and ensure that the probability of each category is bounded away from zero. For the special case of binary observations with $C_j = 2$, the third condition can be ignored. Finally, the assumption on the proportion $\bpi$ states that all $K$ latent classes have comparable cluster sizes, and is included mainly for simpler presentation.

\section{Estimation Method with Entry-wise Guarantees}\label{sec:theory and method}
Given the $N \times J$ observation matrix $\RR$ generated from \eqref{eq:latent class Z}--\eqref{eq:R}, our goal is to estimate the model components:
\begin{enumerate}[(a)]
    \item \label{a}the latent-class vector $\ZZ = (Z_1, \ldots, Z_N)^\top$,
    \item \label{b}the number of blocks $L$ and the partition of the vertices $\{V_1, \ldots, V_L\}$,
    \item \label{c}the covariance/precision matrices $\bSigma, \bOmega$,
    \item \label{d}other continuous parameters (the latent class proportions $\bpi$ and the thresholds $\bDelta$).
\end{enumerate}
For simplicity, we will assume that the number of latent classes $K$ is known throughout this section (see Supplement \ref{supp:select-K} for details on selecting $K$). 

As mentioned in the introduction, our estimation pipeline consists of three steps, which we outline below:
\begin{enumerate}[1.]
    \item Recovering $\ZZ, \bpi, \bDelta$ via spectral clustering (see \cref{subsec:spectral}),
    \item Estimating the covariance matrices $\bSigma$ and the block structure by estimating pairwise polychoric correlations (see \cref{subsec:covariance matrix}),
    \item Estimating the precision matrices $\bOmega$ and finer graphical structures (see \cref{subsec:sparsity}).
\end{enumerate}
In addition to describing the algorithm, each of the first three subsections below provides theoretical finite-sample error bounds, whose proofs are postponed to Supplement \ref{supp:proof}. Further details regarding practical implementation choices are given in \cref{subsec:additional choices}.

\subsection{Step 1: Spectral clustering based on the flattened responses}\label{subsec:spectral}

Spectral clustering is a computationally efficient method widely used in network community detection and mixture models \citep{chen2021spectral,abbe2022l,zhang2024leave}. We leverage the approximate low-rank structure of the $N \times J$ response matrix $\RR$ to recover the latent-class vector $\ZZ$. To see the motivation, first consider the special case where all responses are binary (i.e., $C_j = 2$ for all $j$) and $R_{i,j}$ takes values $0,1$. Then, \eqref{eq:X given Z}--\eqref{eq:R} imply
\begin{equation}\label{eq:conditional expectation}
    \EE\big[R_{i,j} \mid (Z_i = k)\big] = 1 - \Phi(\Delta_{j,1,k}) := \theta_{j,k}.
\end{equation}
In matrix form, we can write $\EE\big[\RR \mid \ZZ \big] = \mathbf{Y} \bTheta^\top$, where $\mathbf{Y}$ is the $N \times K$ matrix whose $i$th row is a one-hot encoding vector of the corresponding scalar latent-class label $Z_i \in [K]$, and $\bTheta = (\theta_{j,k})$ is a $J \times K$ matrix storing the latent-class-specific conditional response probabilities. This shows that the data $\RR$ can be written as the sum of a low-rank matrix $\mathbf{Y} \bTheta^\top$ with rank at most $K$ and a mean-zero noise matrix $\RR - \mathbf{Y} \bTheta^\top$.

In the context of mixture models, spectral clustering is mainly used for real-valued data. In our general setup with potentially $C_j > 2$, one challenge arises from the ordinal nature of the responses. As the response space is non-Euclidean, it is unclear at first glance how to expose the approximate low-rank structure and use it for clustering. We therefore use the data-flattening technique of \cite{chen2024generalized} to transform the data into a \emph{flattened binary matrix}, defining an $N \times \sum_{j = 1}^J C_j$ matrix $\tilde{\RR}$ by
\begin{equation}\label{eq:flatten}
    \tilde{R}_{i, \sum_{m=1}^{j-1} C_m + a} = \mathbb{I}(R_{i,j} = a), \quad \text{for} ~ a \in [C_j].
\end{equation}
In $\tilde{\RR}$, each ordinal entry $R_{i,j} \in [C_j]$ is replaced by a one-hot encoding vector of length $C_j$. For example, when $C_j = 3$, $R_{i,j}=1$ is replaced by $(1,0,0)$, and $R_{i,j}=2$ is replaced by $(0,1,0)$.

After this transformation, the conditional expectation of the flattened data matrix $\tilde{\RR}$ can be spelled out similarly to \eqref{eq:conditional expectation}:
    \begin{align}
        \EE\big[\tilde{R}_{i,\sum_{m=0}^{j-1} C_m + a} \mid (Z_i = k) \big] &= \PP\big(R_{i,j} = a \mid (Z_i = k)\big) \label{eq:conditional expectation extended} \\
        &= \Phi(\Delta_{j,a,k})-\Phi(\Delta_{j,a-1,k}) := \theta_{j,a,k}, \quad \forall j \in [J], ~ a\in [C_j]. \notag
    \end{align}
In matrix form, $\EE[\tilde{\RR} \mid \ZZ] = \mathbf{Y} \bTheta^\top$, where $\bTheta$ is a $\left(\sum_{j = 1}^J C_j\right) \times K$ matrix. This low-rank expectation structure implies that spectral clustering can be applied to $\tilde{\RR}$. The overall procedure is outlined in \cref{algo:spectral}, where we mainly follow the implementation in \cite{zhang2024leave}.

\begin{algorithm}[h]
\SetAlgoLined
 \KwIn{Data matrix $\RR =(\RR_1^\top,\ldots,\RR_N^\top)^\top \in \mathbb{R}^{N\times J}$, number of clusters $K$}
 \KwOut{Cluster assignment vector $\hat{\ZZ} \in [K]^N$}

 \nl Flatten the $N \times J$ ordinal matrix $\RR$ into a $N \times \sum_{j = 1}^J C_j$ binary matrix $\tilde{\RR}$ as in \eqref{eq:flatten}.

 \nl Compute a top-$K$ SVD of $\tilde{\RR}$: $$\tilde{\RR} \approx \hat{\UU} \hat{\bLambda} \hat{\WW}^\top =  \sum_{i=1}^{K} \hat\sigma_i  \hat{\mathbf u}_i  \hat{\mathbf w}_i^\top,$$ where $  \hat \sigma_1 \geq   \hat \sigma_2 \geq \ldots \geq   \hat \sigma_{K} \geq 0$ and $\cbr{\hat{\mathbf u}_i}_{i=1}^{K}\in \mathbb{R}^N, \cbr{\hat{\mathbf w}_i}_{i=1}^{K}\in \mathbb{R}^{\sum_{j=1}^J C_j}$. %

 \nl Perform $k$-means clustering on the rows of $\tilde{\RR} \hat{\WW} = \hat{\UU} \hat{\bLambda}$:
 \begin{align}
 \cbr{\hat{\ZZ},\cbr{\hat{\mathbf d}_k}_{k\in[K]}} = \argmin_{\ZZ\in [K]^N , \cbr{\mathbf d_k}} \sum_{i\in[N]} \|\hat{\WW}^\top \tilde{\RR}_i - \mathbf d_{Z_i}\|^2.\label{eq:k-means}
 \end{align}
\caption{Spectral Clustering of Polytomous LCMs based on Flattening}\label{algo:spectral}
\end{algorithm}

Let $\hat{\ZZ} \in [K]^N$ denote the estimated cluster labels (see \eqref{eq:k-means}), and let $\hat{N}_k = \sum_{i=1}^N \mathbb{I}(\hat{Z}_i = k)$ denote the estimated size of the $k$th cluster. To evaluate the clustering accuracy, define the misclustering error (Hamming loss) as:
\begin{equation}\label{eq:misclustering}
    \ell(\hat{\ZZ}, \ZZ) := \min_{S \in \mathcal{S}_{[K]}} \sum_{i=1}^N \mathbb{I}(\hat{Z}_i \neq S(Z_i)).
\end{equation}
Here, the minimum taken over permutations $\mathcal{S}_{[K]}$ is necessary due to label permutation. Define $\tilde{J}_{\max}$ as the maximum length of a flattened block in $\tilde{\RR}$, and $D$ as the minimum separation between the threshold vectors $\bDelta_k$:
$$\tilde{J}_{\max} := \max_{\ell=1}^L \sum_{j \in V_\ell} C_j, \quad D := \min_{k \neq k' \in [K]} \|\bDelta_k - \bDelta_{k'}\|.$$ %
The following proposition provides exact recovery guarantees for the latent-class vector $\ZZ$ in the high-dimensional regime. Note that all dimensions $N,J,\tilde{J}_{\max},K$ are allowed to grow, while $\tilde{J}_{\max}$ and $K$ may also be of constant order.

\begin{proposition}\label{prop:clustering}
    Suppose that Assumptions \ref{assmp:bounded number of categories}, \ref{assmp:bounded parameter} hold. %
    There exist constants $c_1, c_2 > 0$ (that depend on $C_{\max}, \delta_3, M$) such that when $\min\left(D, \sqrt{K} \sigma_K(\bTheta)\right) \ge c_1 K\sqrt{\tilde{J}_{\max}}\left(1+\sqrt{J/N} \right)$ and $c_1 N \ge K^2$, we have
    \begin{equation}\label{eq:clustering accuracy}
        \EE \ell(\hat{\ZZ}, \ZZ) \le N \Big(e^{-{c_2 D^2}/{\tilde{J}_{\max}}} + Ke^{-{c_2N}/{K}}\Big) {=: p_1}.
    \end{equation}
    In particular, when $D \ge c_3 \sqrt{\tilde{J}_{\max} \log N}$ for some large enough $c_3>0$, the right hand side of \eqref{eq:clustering accuracy} is $O(N^{-1})$, so $\hat{\ZZ}$ achieves \textbf{exact recovery} with probability at least $1-O(N^{-1})$. 
\end{proposition}

We discuss the conditions required in \cref{prop:clustering}. First, the lower bound for $D$ requires all cluster means to be sufficiently separated. Note from \eqref{eq:conditional expectation extended} that the thresholds $\bDelta_k$ (which are used to define $D$) are directly related to the cluster centers $\btheta_k$. Second, the lower bound for $\sigma_K(\bTheta)$ requires a sufficiently large spectral gap, and is a standard requirement for the success of spectral methods. The condition $N \gtrsim K^2$ is a mild technical requirement that arises from \cite{zhang2024leave}.

The following simplified scaling discussion is provided only for illustration and interpretation; the additional conditions imposed in this paragraph are not required for the conclusion of \cref{prop:clustering} to hold. To further understand the scaling requirement, suppose $K = O(1)$. Also, assume homogeneous block sizes and response categories (i.e., $J_1 = \ldots = J_L \equiv J/L, C_1 = \ldots = C_J \equiv C$), so that $\tilde{J}_{\max} = O(J/L)$. As $\bDelta_k$ is an $O(J)$-dimensional parameter, we can view $D$ as $O(\sqrt{J})$ when the latent classes are distinct across all items. Then, the requirement for $D$ simplifies to assuming a \emph{growing number of blocks} so that $L \gtrsim \log N + J/N$. Taking $\sigma_K(\bTheta) \asymp \sqrt{J}$, the requirement for $\sigma_K(\bTheta)$ also reduces to $L \gtrsim 1+J/N$. Note that these requirements are relatively mild and allow both $J > N$ as well as $N > J$. In general, the requirement for $L$ becomes more stringent when different latent classes are less separated or the number of clusters $K$ increases. %

Now, given the estimated cluster labels, it is straightforward to estimate the mixture proportion vector $\bpi$ as well as the threshold vector $\bDelta$. To simplify notation, assume without loss of generality that the estimated labels are correctly permuted in the sense that $\hat{Z}_i = Z_i$. The latent class proportions can be estimated based on the latent class labels, $\hat{\pi}_k := \hat{N}_k/N.$ 
Recalling \eqref{eq:R}, the conditional probability of observing a response of $a$ or lower is $\PP(R_{j} \le a \mid Z = k) = \Phi(\Delta_{j,a,k})$. A natural moment estimator arises by replacing the left hand side $\PP(R_{j} \le a \mid Z = k)$ of this population identity with the empirical probability conditional on the $k$th latent class:
$$\hat{\Delta}_{j,a,k} := \Phi^{-1} \left( \frac{1}{\hat{N}_k} \sum_{i: \hat{Z}_i=k} \mathbb{I}(R_{i,j} \le a) \right), \quad \forall a \in [C_j-1].$$
This estimator can also be obtained from the estimated $K$-dimensional centers $\hat{d}_k$ in \eqref{eq:k-means} by mapping them back to the original $J$-dimensional scale and writing out the k-means criterion; see eq. (15) in \cite{zhang2024leave} for further details.

The following proposition gives entry-wise estimation error bounds for both $\hat{\bpi}$ and $\hat{\bDelta}$.

\begin{proposition}\label{prop:pi delta consistency}
    Assuming all conditions in \cref{prop:clustering}, there exist constants $c_1, c_2 = c_2(M,\delta_3)>0$ such that 
    \begin{align*}
        \PP\Big(\max_{k \in [K]}|\hat{\pi}_k - \pi_{k}| \ge t\Big) &\le Ke^{- c_1 N t^2} + O(N^{-1}), \\
        \PP\Big(\max_{k \in [K]} \max_{j \in [J], a \in [C_j-1]} |\hat{\Delta}_{j,a,k} - \Delta_{j,a,k}| \ge t \Big) &\le \Big(\sum_{j=1}^J C_j\Big) K e^{-c_2 N t^2/K} +O(N^{-1}),
    \end{align*}
    for all $0 < t < M$. %
\end{proposition}
\noindent In particular, both estimators for $\bpi$ and $\bDelta$ are $\sqrt{N}$-consistent, up to the log factors in $K, J$ that arise from a uniform bound.
This result suggests that the proposed spectral clustering method based on the flattened data matrix delivers an entrywise consistent estimator for the latent class proportions and the threshold values for the ordinal variables, laying the foundation for subsequent steps.

\subsection{Step 2: Estimating covariance matrix and block structures}\label{subsec:covariance matrix}
Our second step aims to recover the block dependence structure by partitioning the set of vertices $V = [J]$ into disjoint sets $V_1, \ldots, V_L$. The key idea is to separately estimate the class-specific latent covariance matrices $\bSigma_k$ for each latent class $k$, and aggregate their sparsity patterns to estimate the partition shared across latent classes.

We first consider estimating the covariance matrices $\bSigma_k$ for each cluster. If the latent Gaussian variables $X_{i,j}$ were observed, one could simply use the sample covariance matrix. However, with ordinal responses, this is not feasible. A naive solution would be to compute the MLE of $\bSigma_k$ given the responses within the $k$th cluster. %
However, computing this is challenging as $\bSigma_k$ has $O(J^2)$ entries and is required to be positive definite.
Therefore, we relax the positive definite constraint and estimate each off-diagonal entry separately. Our key observation is that $\sigma_{(j,j'),k}$ is exactly the \emph{polychoric correlation} between observations $R_j$ and $R_{j'}$, conditioned on the latent class $k$. To the statistical audience, the terminology ``tetrachoric correlation'' \citep{pearson1900mathematical} may be more familiar, which is a special case of the polychoric correlation when both $R_j$ and $R_{j'}$ are binary. Estimating pairwise polychoric correlation from ordinal responses has been a classical topic in psychometrics and categorical data analysis \citep{lancaster1964estimation,olsson1979maximum,muthen1984general,bonett2005inferential}. %
Motivated by \cite{olsson1979maximum}, we focus on the following pseudo-MLE that maximizes the likelihood of $(R_j, R_{j'})$ while treating the threshold vector $\bDelta_k$ as fixed:
\begin{align}\label{eq:sigma hat}
    \hat{\sigma}_{(j,j'),k} :&= \argmax_{|\sigma| \le 1 - \delta_1} \sum_{i: \hat{Z}_i = k} \log \PP(R_{j}^{(\sigma)} = R_{i,j}, R_{j'}^{(\sigma)} = R_{i,j'} \mid Z=k ; \hat{\bDelta}_k) \\
    &= \argmax_{|\sigma| \le 1 - \delta_1} \sum_{i: \hat{Z}_i = k} \log \PP\Big(X_{j}^{(\sigma)} \in [\hat{\Delta}_{j,R_{i,j}-1,k}, \hat{\Delta}_{j,R_{i,j},k}), X_{j'}^{(\sigma)} \in [\hat{\Delta}_{j',R_{i,j'}-1,k}, \hat{\Delta}_{j',R_{i,j'},k}) \Big). \notag
\end{align}
Here, $(X_{j}^{(\sigma)}, X_{j'}^{(\sigma)})$ is the latent Gaussian vector, distributed as a bivariate standard normal with correlation $\sigma.$ Also, $R_{j}^{(\sigma)}$ denotes the random variable discretized from $X_{j}^{(\sigma)}$ via \eqref{eq:R}, and $R_{i,j}$ is the observation from the $i$th individual.

\begin{remark}[Connection to rank-based estimators in semiparametric statistics]\label{rmk:rank}
    There are alternative estimators in the ordinal Gaussian copula literature that estimate pairwise latent correlations by inverting Kendall's-tau moment equations: \cite{fan2017high} considers the binary responses case, and \cite{feng2019high} and \cite{quan2018rank} consider the ordinal case. When all variables are binary, \eqref{eq:sigma hat} is essentially equivalent to the Kendall's tau estimator. %
    This equivalence no longer holds when at least one response has $C_j> 2$ categories, which is why \cite{feng2019high} considered an ensemble estimator that aggregates estimates based on binarized responses and \cite{quan2018rank} limits to ternary data (with $C_j=3$), both of which result in information loss for general ordinal data. In contrast, the likelihood-based estimator \eqref{eq:sigma hat} we employ uses the full $C_j\times C_{j'}$ contingency table to estimate the correlation and allows arbitrary numbers of categories per variable.
\end{remark}

\begin{remark}[Fixing the threshold estimates]
    The careful reader may have noticed that the threshold vectors $\bDelta_k$ are treated as constants in the likelihood \eqref{eq:sigma hat}. This is because $\bDelta_k$ are already estimated consistently in Step 1, and because joint maximum likelihood over $(\bDelta_k,\bSigma_k)$ is computationally prohibitive even for moderate $J$ \citep{olsson1979maximum,welz2024robust}. Without fixing $\bDelta_k$, pairwise estimation is no longer possible, since each $\bDelta_{j,k}$ appears in the likelihood for all $\{\sigma_{(j,j'),k}: j' \neq j\}.$ A similar strategy was used in the works cited in \cref{rmk:rank}.

\end{remark}

Now, %
we can ensure that each entry of the covariance matrices $\bSigma_k$ is accurately estimated, at a $\sqrt{N/K}$ rate. Here, $N/K$ is the effective within-class sample size, due to \cref{assmp:bounded parameter}. 
In fact, taking $K=1$, the resulting rate (including the log factor) matches the usual estimation error rate for Gaussian graphical models and semiparametric copulas \citep{liu2012high}, despite observing only discrete data. The proof of \cref{prop:covariance estimation} builds upon an exponential concentration inequality for the pairwise polychoric correlation estimator \eqref{eq:sigma hat}, which may be of independent interest (see \cref{lem:polychoric} in the Supplementary Material).

\begin{proposition}\label{prop:covariance estimation}
    Consider the setting of \cref{prop:clustering}.
        For any $t>0$, there exist constants $c_1,c_2,c_3>0$ (that depend on $C_{\max}, \delta_1, \delta_2, \delta_3, M$) such that the following holds:
        \begin{align}\label{eq:tail sigma}
            \PP\Big(\max_{k=1}^K \|\hat{\bSigma}_k - \bSigma_k\|_{\max} \ge t\Big) \le c_1 K J^2 \Big(e^{-c_2 N t^2/K} + e^{-c_3 N/K}\Big) + p_1.
        \end{align}
        Hence, there exists another constant $c_4>0$ such that $\max_{k=1}^K \|\hat{\bSigma}_k - \bSigma_k\|_{\max} \le \tau_1 := c_4 \sqrt{K\log (JK)/N}$ with probability greater than $1-O\left((JK)^{-1} + N^{-1} \right)$ provided that $N \gtrsim K\log(JK)$. 
\end{proposition}

Consequently, we can also learn the global block structure among $X_1, \ldots, X_J$, i.e., estimate $\hat{L}$ and the partition $V = [J] = \cup_{\ell=1}^{\hat{L}} \hat{V}_\ell$. Our main idea is to aggregate the sparsity information across all latent classes, as a single latent class may potentially exhibit additional sparsity and therefore a finer block structure than the shared global block structure. See the top row in \cref{fig:block recovery} for a visualization of $\bSigma_1$ and $\hat{\bSigma}_1$, where the block structure is not evident from $\hat{\bSigma}_1$ alone. To this end, 
we define a $J \times J$ binary matrix $\GG$ by setting each entry as $g_{(j,j')} := \mathbb{I}(\max_{k=1}^K |\hat{\sigma}_{(j,j'),k}| > \tau_1),$ and view it as an adjacency matrix on the set of nodes $V$. By \cref{prop:covariance estimation}, $\GG$ includes the edges $\cup_{\ell = 1}^L E_\ell$, where $E_\ell$ is as in \cref{assump:minimal strength}. Then, we can read off the block structure from $\GG$ by setting $\hat{L}$ as the number of connected components, and letting $\hat{V}_1,\ldots,\hat{V}_{\hat{L}}$ be vertices in each connected component. The following proposition establishes consistency of this estimator.

\begin{proposition}\label{prop:block structure}
    Continuing from \cref{prop:covariance estimation}, additionally suppose that the constant $\epsilon_N$ in \cref{assump:minimal strength} satisfies $\epsilon_N \ge 2 \tau_1$.
    Then, $\hat{L} = L$ and $\{\hat{V}_1,\ldots,\hat{V}_{\hat{L}}\} = \{V_1,\ldots,V_L\}$ with probability greater than $1-O\left((JK)^{-1} + N^{-1} \right)$.
\end{proposition}

Given this result, we can now fix the global block structure. For notational convenience and without loss of generality, on the event that $V$ is correctly partitioned, assume that the learned blocks are indexed correctly; i.e., $\hat{V}_\ell = V_\ell$ for all $\ell$.

\subsection{Step 3: Learning precision matrix and within-block structures}\label{subsec:sparsity}
Having estimated the ``global'' block structure shared by all latent classes, we now consider the problem of estimating the ``local'' sparsity structure encoded in each precision matrix $\bOmega_k^{(\ell)}$ for each block $\ell$ and each latent class $k$. Note that imposing the block structure greatly reduces the dimension of the precision matrices from ${J(J+1)}/2$ to $\sum_{\ell=1}^L J_{\ell}(J_{\ell}+1)/2$. Since no common structure is imposed within blocks, each $\bOmega_k^{(\ell)}$ is separately estimated. 

Here, the main idea is to use sparse precision matrix estimators from their covariance matrix counterparts $\bSigma_k^{(\ell)}$. For this purpose, for any $k \in [K], \ell \in [L]$, one can use the usual graphical lasso estimator \citep[][among others]{friedman2008sparse, yuan2007model}: %
\begin{align*}%
    \hat{\bOmega}_k^{(\ell)} = \argmax_{\bOmega_k^{(\ell)} \in \R^{|\hat{V}_\ell|\times |\hat{V}_\ell|}} \log|\bOmega_k^{(\ell)}| - \text{tr}(\hat{\bSigma}_k^{(\ell)} \bOmega_k^{(\ell)}) - \lambda \sum_{j \neq j' \in \hat{V}_\ell} |\omega_{(j,j'),k}^{(\ell)}|,
\end{align*}
or the CLIME estimator \citep{cai2011constrained}:
\begin{align}
    \tilde{\bOmega}_k^{(\ell)} &= \argmin_{\bOmega_k^{(\ell)} \in \R^{|\hat{V}_\ell|\times |\hat{V}_\ell|}} \|\bOmega_k^{(\ell)}\|_1 \quad \text{s.t.} \quad \|\hat{\bSigma}_k^{(\ell)} \bOmega_k^{(\ell)} - \II\|_{\max} \le \lambda, \label{eq:clime initial}\\
    \hat{\omega}_{(j,j'),k} &= \hat{\omega}_{(j',j),k} :=
    \begin{cases}
        \tilde{\omega}_{(j,j'),k}, & \text{if } |\tilde{\omega}_{(j,j'),k}| \le |\tilde{\omega}_{(j',j),k}|,\\
        \tilde{\omega}_{(j',j),k}, & \text{otherwise},
    \end{cases}
    \quad \forall j,j' \in [J].\label{eq:precision matrix estimation}
\end{align}
Here, the initial estimator in \eqref{eq:clime initial} may be asymmetric, so this is symmetrized in \eqref{eq:precision matrix estimation} by taking the entry with smaller magnitude. 

For a simpler presentation, here we focus on analyzing the CLIME estimator \eqref{eq:precision matrix estimation}. %
For this purpose, we require additional assumptions on the true precision matrix. %
\begin{assumption}[Precision matrix assumptions]\label{assmp:precision matrix}

\begin{enumerate}[(a)]
    \item For each $k \in [K], \ell \in [L]$, suppose that $\bOmega_{k}^{(\ell)}$ %
    satisfies $\|\bOmega_{k}^{(\ell)}\|_1 \le M_N'$ for some $M_N'>0$. %
    \item Additionally, for some constant $\epsilon_N'>0$, assume a minimal signal strength condition $$\min_{j,j' \in V_\ell: ~\omega_{(j,j'),k}^{(\ell)} \neq 0} |\omega_{(j,j'),k}^{(\ell)}| > \epsilon_N'.$$
\end{enumerate}
\end{assumption}
\noindent \cref{assmp:precision matrix} is a standard requirement in \cite{cai2011constrained} that helps theoretically understand the CLIME estimator. Part (b) in \cref{assmp:precision matrix} requires a minimal signal strength for the nonzero precision matrix entries as opposed to the analogous condition for selecting covariance matrix entries in \cref{assump:minimal strength}.

Now, we state our main theoretical result on precision matrix estimation. Again, the resulting bounds are of the same order as those for Gaussian graphical models.
\begin{proposition}\label{prop:precision matrix}
    Consider the setting in \cref{prop:block structure} and suppose \cref{assmp:precision matrix}(a) holds. Consider the precision matrix estimate $\hat{\bOmega}_k^{(\ell)}$ in \eqref{eq:precision matrix estimation} with $\lambda = M_N' \tau_1$.
    
    \begin{enumerate}[(a)]
        \item With probability $1-O\left((JK)^{-1} + N^{-1}\right)$, we have 
        $$\|\hat{\bOmega}_{k}^{(\ell)} - \bOmega_k^{(\ell)}\|_{\max} \le 4 M_N' \lambda.$$

        \item Additionally, let \cref{assmp:precision matrix}(b) hold for some $\epsilon_N' \ge 8 M_N' \lambda$. Then, for $\tau_2 = 4M_N'\lambda$, the thresholded estimator $\tilde{\bOmega}_{k}^{(\ell)}$ with
        $$\tilde{\omega}_{(j,j'),k}^{(\ell)} := \hat{\omega}_{(j,j'),k}^{(\ell)} \mathbb{I}(|\hat{\omega}_{(j,j'),k}^{(\ell)}| \ge \tau_2)$$
        is sign consistent, i.e., $\sgn(\tilde{\omega}_{(j,j'),k}^{(\ell)}) = \sgn({\omega}_{(j,j'),k}^{(\ell)})$ with probability $1-O\left((JK)^{-1} + N^{-1}\right)$.
    \end{enumerate}
\end{proposition}

\begin{remark}
    Similar sign-consistency guarantees can potentially be derived for other sparsity-inducing procedures, such as the graphical lasso, by applying the so-called ``meta-theorem'' (cf. Thm 4.3 in \cite{liu2012high}) to the entry-wise guarantees for $\bSigma_k$ in \cref{prop:covariance estimation}. To this end, one would need to replace \cref{assmp:precision matrix} with the usual Lasso assumptions, such as mutual incoherence.
\end{remark} %

Finally, we combine all results in Propositions \ref{prop:clustering}--\ref{prop:precision matrix} to obtain the following theorem.
\begin{theorem}
    Under the conditions of Propositions \ref{prop:clustering}--\ref{prop:precision matrix}, which include the scaling requirement $N \gtrsim K^2 + K \log(JK)$, all model components \ref{a}--\ref{d} can be consistently estimated. %
\end{theorem}
\noindent Thus, our three-step estimation procedure provides end-to-end consistency for estimating all model components in the block-dependent latent class model. 

\subsection{Additional choices for implementation}\label{subsec:additional choices}

\paragraph{Learning the block structures}
While the binary matrix $\GG$ (see the text just before \cref{prop:block structure}) suffices for theoretical purposes, it can lead to false discovery under limited sample sizes (i.e., some zero values in the off-block-diagonals are estimated as one). Such errors can blur the connectivity structure between the blocks, as illustrated in the left panel in the second row of \cref{fig:block recovery}. 

To provide a more robust and practical solution for learning the shared block structure, we propose using network community detection on a continuous aggregation matrix. Specifically, we aggregate the \emph{magnitudes} of the estimated polychoric correlations (as opposed to the binary matrix $\GG$) and define a $J \times J$ weighted adjacency matrix $\tilde{\GG}$ with zero diagonals:
\begin{align}\label{eq:g tilde}
    \tilde{g}_{j,j'} := 
    \begin{cases}
        \frac{1}{K} \sum_{k=1}^K |\hat{\sigma}_{(j,j'),k}|\mathbb{I}(|\hat{\sigma}_{(j,j'),k}| \ge \tau_1) & \text{if } j \neq j', \\
        0 & \text{if } j = j'.
    \end{cases}
\end{align}
Here, $\tilde{\GG}$ can be made tuning-free by not thresholding the estimates $\hat{\sigma}_{(j,j'),k}$. Treating $\tilde{\GG}$ as a weighted adjacency matrix on $V=[J]$, we can estimate the shared block structure by generic graph community detection. This is because the off-block-diagonal entries are mostly small by \cref{prop:covariance estimation}, while within-block entries form connected components.

\begin{figure}[h!]
    \centering
    \includegraphics[width=0.8\linewidth]{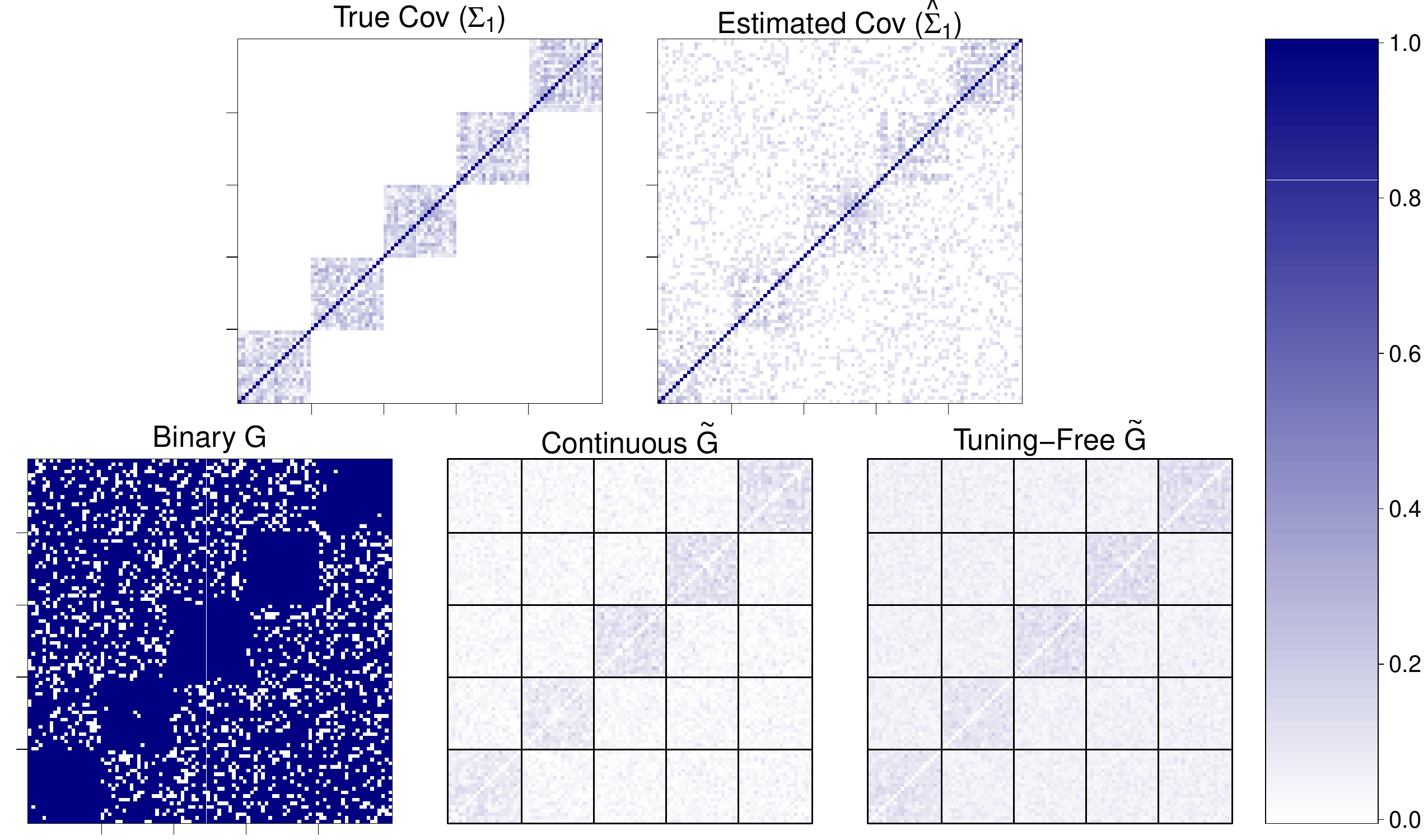}
    \caption{Covariance matrices and block-recovery for a simulated dataset (under the setting of the first row in \cref{tab:sim result}). Top: $|\bSigma_1|$ (left) and $|\hat{\bSigma}_1|$ (right). Bottom: aggregation matrices for block structure recovery; the binary matrix $\GG$ (left), the continuous matrix $\tilde{\GG}$ \eqref{eq:g tilde} (center), and the tuning-free version of $\tilde{\GG}$ (right). 
    Black lines mark the estimated blocks. The continuous matrices better capture the block signal despite the off-block noise. %
   }
    \label{fig:block recovery}
\end{figure} %

For community detection, we use the Leiden algorithm \citep{traag2019louvain} because it can estimate $L$ and is well-suited for weighted graphs. In the center/right panel of the second row in \cref{fig:block recovery}, Leiden accurately recovers the blocks under both versions of $\tilde{\GG}$.  This choice can be flexibly substituted with other methods, such as the Louvain algorithm \citep{blondel2008fast} or spectral clustering. The full procedure is summarized in \cref{algo:learn shared sparsity}.

\begin{algorithm}[h!]
     \SetAlgoLined
     \KwIn{Estimated covariance matrices $\{\hat{\bSigma}_{k}\}_{k \in [K]}$}
     \KwOut{Estimated number of blocks $\hat{L}$ and a partition $\{\hat{V}_\ell\}_{\ell \in [\hat{L}]}$ of the variables $[J]$}
     \nl Construct the continuous weighted adjacency matrix $\tilde{\GG}$ as defined in \eqref{eq:g tilde}.

     \nl Apply a community detection algorithm (e.g., the Leiden algorithm) on $\tilde{\GG}$. %
     
     \nl Let $\hat{L}$ be the automatically detected number of communities, and let $\{\hat{V}_\ell\}_{\ell \in [\hat{L}]}$ be the corresponding partition of the variables.
     \caption{Learning the shared block structure}\label{algo:learn shared sparsity}
\end{algorithm}

The Leiden algorithm also has a \emph{resolution parameter} that controls the granularity of the clustering, where a higher value produces a finer-grained partition with more clusters. In simulations, we use the default resolution parameter. In the real data analysis, we also consider a higher resolution parameter to find finer-grained blocks.

\paragraph{Selecting the tuning parameter $\lambda$}
The tuning parameter $\lambda$ in \eqref{eq:clime initial} controls the level of sparsity in the estimated precision matrix. To select an optimal $\lambda$, we employ 5-fold cross-validation, a standard choice for sparse Gaussian graphical models \citep{rothman2008sparse}. 
Information criteria such as BIC or EBIC \citep{chen2008extended,foygel2010extended} are alternatives, but in our experiments they tended to select overly sparse models when $J$ was large, say above 100.

Note that one may further fine-tune $\lambda$ for each cluster/block, by considering individual penalties $\lambda_{k}^{(\ell)}$ for each $\bOmega_k^{(\ell)}$. This could allow varying levels of sparsity across clusters/blocks.

\section{Simulation Studies}\label{sec:sims}
We empirically assess the proposed method under the following simulation settings:
\begin{itemize}
    \item \textbf{Dimensions}: We consider a fixed sample size $N = 1000$ and vary (i) the number of variables $J \in \{100, 500\}$, (ii) the number of true blocks $L \in \{5, 10\}$, and (iii) the number of latent classes $K \in \{5, 10\}$. 
    
    \item \textbf{Number of categories}: The first half of the $J$ items are binary responses (i.e., $C_j=2$), and the second half are four-category ordinal responses (i.e., $C_j=4$).
    
    \item \textbf{Graph structure}: We consider two block-diagonal graph structures with varying sparsity of (iv) $s = 0.2$ and $s = 0.6$. The graph within each block follows the Erd\H{o}s-R\'enyi random graph with probability $s$, i.e., each off-diagonal entry within a block is non-zero with probability $s$ and $0$ with probability $1-s$.
\end{itemize}
Data generation details are provided in Supplement \ref{supp:sim data generation}. Note that the range of $J$ is high-dimensional in graphical modeling, as estimating $\bOmega = (\bOmega_1, \ldots, \bOmega_K)$ entails $KJ^2$ entries.
For each setting, we run 100 replications and report the following averaged metrics:
\begin{itemize}
    \item For the latent class assignments, we report (i) $\text{Err}(\hat{\ZZ})$, the misclustering error in \eqref{eq:misclustering}.
    \item For the block recovery, we report (ii) $\text{Err}(\hat{L}) = \mathbb{I}(\hat{L} \neq L)$, i.e. the 0-1 loss for the estimated number of blocks. We also report the Adjusted Rand Index (ARI) \citep{hubert1985comparing} between the partitions $\{V_1, \ldots, V_L\}$ and $\{\hat{V}_1,\ldots,\hat{V}_{\hat{L}}\}$, and denote it as (iii) ARI$(\hat{V})$. The value of 1 indicates perfect clustering and a near-zero value indicates uninformative (random) clustering.

    \item For the precision matrices $\bOmega$, we report the (iv) support recovery (sparsity) error and (v) Frobenius error: %
    \begin{align*}
        \text{Err}_0(\hat{\bOmega}) &:= \frac{1}{K J^2} \sum_{k=1}^K \sum_{j,j'=1}^J \mathbb{I}\Big(\sgn(\hat{\omega}_{(j,j'),k}) \neq \sgn(\omega_{(j,j'),k})\Big), \\
    \text{Err}_F(\hat{\bOmega}) &:= \frac{1}{\sqrt{K J^2}}\sqrt{ \sum_{k=1}^K \|\hat{\bOmega}_{k} - \bOmega_{k}\|_F^2}.
    \end{align*}
\end{itemize}

\begin{table}[h!]
\centering
\begin{tabular}{cccc|ccccc|c}
\toprule
    $J$ & $K$ & $L$ & $s$ & Err($\hat{\ZZ}$) & Err($\hat{L}$) & 1-ARI($\hat{V}$) & $\text{Err}_0(\hat{\bOmega})$ & $\text{Err}_F(\hat{\bOmega})$ & time (s) \\
    \midrule
    \multirow{8}{*}{$100$}   & \multirow{4}{*}{$5$}  & \multirow{2}{*}{$5$}  & 0.6 & 0.00 & 0.00 & 0.00 & 0.084 & 0.103 & 2.4 \\
                            &                       &                       & 0.2 & 0.00 & 0.00 & 0.00 & 0.085 & 0.099 & 2.3 \\
                            &                       & \multirow{2}{*}{$10$} & 0.6  & 0.00 & 0.18 & 0.02 & 0.035 & 0.086 & 2.5 \\
                            &                       &                       & 0.2 & 0.00 & 0.41 & 0.06 & 0.038 & 0.084 & 2.3 \\
                            & \multirow{4}{*}{$10$} & \multirow{2}{*}{$5$}  & 0.6  & 0.05 & 0.00 & 0.00 & 0.090 & 0.113 & 4.1 \\
                            &                       &                       & 0.2 & 0.06 & 0.00 & 0.00 & 0.068 & 0.113 & 3.9 \\
                            &                       & \multirow{2}{*}{$10$} & 0.6  & 0.02 & 0.57 & 0.08 & 0.045 & 0.097 & 4.1 \\
                            &                       &                       & 0.2 & 0.04 & 0.67 & 0.09 & 0.038 & 0.097 & 4.1 \\
    \midrule
    \multirow{8}{*}{$500$}  & \multirow{4}{*}{$5$}  & \multirow{2}{*}{$5$}  & 0.6  & 0.00 & 0.00 & 0.00 & 0.102 & 0.060 & 47.0 \\
                            &                       &                       & 0.2 & 0.00 & 0.00 & 0.00 & 0.087 & 0.059 & 46.5 \\
                            &                       & \multirow{2}{*}{$10$} & 0.6  & 0.00 & 0.00 & 0.01 & 0.052 & 0.058 & 46.5 \\
                            &                       &                       & 0.2 & 0.00 & 0.00 & 0.00 & 0.027 & 0.056 & 46.2 \\
                            & \multirow{4}{*}{$10$} & \multirow{2}{*}{$5$}  & 0.6  & 0.00 & 0.00 & 0.00 & 0.106 & 0.067 & 83.7 \\
                            &                       &                       & 0.2 & 0.00 & 0.00 & 0.00 & 0.098 & 0.062 & 85.7 \\
                            &                       & \multirow{2}{*}{$10$} & 0.6  & 0.00 & 0.02 & 0.00 & 0.051 & 0.058 & 83.2 \\
                            &                       &                       & 0.2 & 0.00 & 0.00 & 0.00 & 0.033 & 0.058 & 84.3 \\
\bottomrule
\end{tabular}
\caption{Simulation results under varying dimensions $J,K,L$ and sparsity $s$. The metrics are averaged over 100 replications, where the first column takes values in $[0, N]$ with $N = 1000$, the second--fourth column takes values in $[0,1]$, and the fifth column (the scaled Frobenius norm) is unbounded. For all columns, \textbf{smaller values are better}.
\label{tab:sim result}
}
\end{table}

\cref{tab:sim result} illustrates that the proposed method consistently achieves high accuracy in all of clustering, block-structure recovery, and precision matrix estimation. Regarding clustering, exact clustering is possible except when $J = 100, K = 10$, and the clustering error is small even in this challenging setting. This aligns with theoretical insights from \cref{prop:clustering} that clustering is harder when $J$ decreases and $K$ increases. In Supplement \ref{sec:em comparison}, we additionally discuss a one-step likelihood refinement procedure to further improve clustering accuracy.

Next, for block structure recovery, our method successfully identifies the true number of blocks ($\hat{L} = L$) and achieves a high ARI for the estimated partition $\{\hat{V}_1, \ldots, \hat{V}_{\hat L}\}$ in most of settings. Even when the point estimate $\hat L$ is incorrect, the high ARI (greater than 0.9) indicates that the resulting partitions are near-correct and informative; see Supplement \ref{sec:robust} for a visualization. We also observe that when $\hat L$ is incorrect, it is typically underestimated, i.e., $\hat{L} \le L$. This may be an inherent limitation of modularity-based graph clustering algorithms when the block sizes $J/L$ are small \citep{fortunato2007resolution,traag2015detecting}, and we believe that this issue may be mitigated by fine-tuning the resolution parameter. See Supplement \ref{supp:sensitivity} for further illustration regarding the resolution parameter.

Given the near-perfect recovery of the blocks, the precision matrices $\bOmega$ are also well estimated in terms of both the support and values. Overall, all reported accuracy measures are comparable as the sparsity level $s$ changes, illustrating that our pipeline can handle unknown, varying sparsity. We will further discuss precision matrix estimation in the context of the comparison studies that will follow. Finally, the runtime reported in the last column of \cref{tab:sim result} demonstrates the scalability of the method. %

\paragraph{Comparison and ablation studies}
We compare Steps 2--3 of our pipeline with the joint graphical lasso (JGL) from \cite{danaher2014joint}. After Step 1 and estimating the covariance matrices $\hat{\bSigma}_k$ in Step 2, JGL is a natural alternative because its joint penalty (across all $k$) helps learn a common block-diagonal structure for the precision matrices, followed by within-block estimation. In fact, \cite{danaher2014joint} proposed an explicit criterion that helps characterize block-diagonal structures and used it to scale up computation.
We fix $J=100$ and use group JGL, which encourages shared sparsity across all $K$ precision matrices.

\begin{table}[h!]
\centering
\resizebox{\textwidth}{!}{
\begin{tabular}{ccc|ccc|ccc|ccc|ccc}
\toprule
\multirow{2}{*}{$K$} & \multirow{2}{*}{$L$} & \multirow{2}{*}{$s$}
& \multicolumn{3}{c|}{Err($\hat{L}$)}
& \multicolumn{3}{c|}{1-ARI($\hat{V}$)}
& \multicolumn{3}{c|}{$\text{Err}_F(\hat{\bOmega})$}
& \multicolumn{3}{c}{time (s)} \\
& & 
& Ours & JGL-O & JGL-CV
& Ours & JGL-O & JGL-CV
& Ours & JGL-O & JGL-CV
& Ours & JGL-O & JGL-CV \\
\midrule
\multirow{4}{*}{$5$}
& \multirow{2}{*}{$5$}
& 0.6 & 0.00 & 0.80 & 1.00 & 0.00 & 0.88 & 1.00 & 0.103 & 0.126 & 0.121 & 0.1 & 0.6 & 0.5 \\
& & 0.2 & 0.00 & 0.70 & 1.00 & 0.00 & 0.36 & 1.00 & 0.099 & 0.133 & 0.124 & 0.1 & 0.5 & 0.4 \\
& \multirow{2}{*}{$10$}
& 0.6 & 0.18 & 0.89 & 1.00 & 0.02 & 0.65 & 1.00 & 0.086 & 0.111 & 0.105 & 0.1 & 0.6 & 0.4 \\
& & 0.2 & 0.41 & 0.74 & 1.00 & 0.06 & 0.21 & 1.00 & 0.084 & 0.116 & 0.108 & 0.1 & 0.5 & 0.4 \\
\midrule
\multirow{4}{*}{$10$}
& \multirow{2}{*}{$5$}
& 0.6 & 0.00 & 0.84 & 1.00 & 0.00 & 1.00 & 1.00 & 0.113 & 0.132 & 0.132 & 0.1 & 1.1 & 0.9 \\
& & 0.2 & 0.00 & 0.87 & 1.00 & 0.00 & 0.86 & 1.00 & 0.113 & 0.141 & 0.140 & 0.1 & 1.0 & 0.9 \\
& \multirow{2}{*}{$10$}
& 0.6 & 0.57 & 0.89 & 1.00 & 0.08 & 0.96 & 1.00 & 0.097 & 0.117 & 0.117 & 0.1 & 1.2 & 0.9 \\
& & 0.2 & 0.67 & 0.91 & 1.00 & 0.09 & 0.55 & 1.00 & 0.097 & 0.125 & 0.124 & 0.1 & 0.9 & 1.0 \\
\bottomrule
\end{tabular}
}
\caption{Simulation results for $J=100$ while varying $K,L,s$. `Ours' denotes the proposed method, and `JGL-O'/`JGL-CV' denotes the joint graphical lasso from \cite{danaher2014joint}, where the tuning parameter is fine-tuned using the true $L$ in an oracle manner, or selected via cross-validation. For all columns, \textbf{smaller values are better}. Runtime reports the final comparison-stage fit (after estimating $\hat{\bSigma}$) and excludes initial tuning.}
\label{tab:sim comparison jgl}
\end{table}

\cref{tab:sim comparison jgl} shows that the proposed method more accurately recovers the block structure and estimates the precision matrices. The `JGL-O' column fixes $\lambda_2 = 0.001$ and chooses $\lambda_1$ from a grid with increment $0.001$ so that the average $\hat{L}$ best matches the true number of blocks $L$, which is unavailable in practice. Even with this oracle-style tuning, JGL results in inconsistent block recovery in all settings. Also, JGL-CV always selects $\hat{L}=1$ (no blocks detected), which indicates the need for careful tuning (see Supplement \ref{supp:sensitivity} for additional discussion on sensitivity). The proposed method is also computationally favorable. %

\begin{figure}
    \centering
    \includegraphics[width=\linewidth]{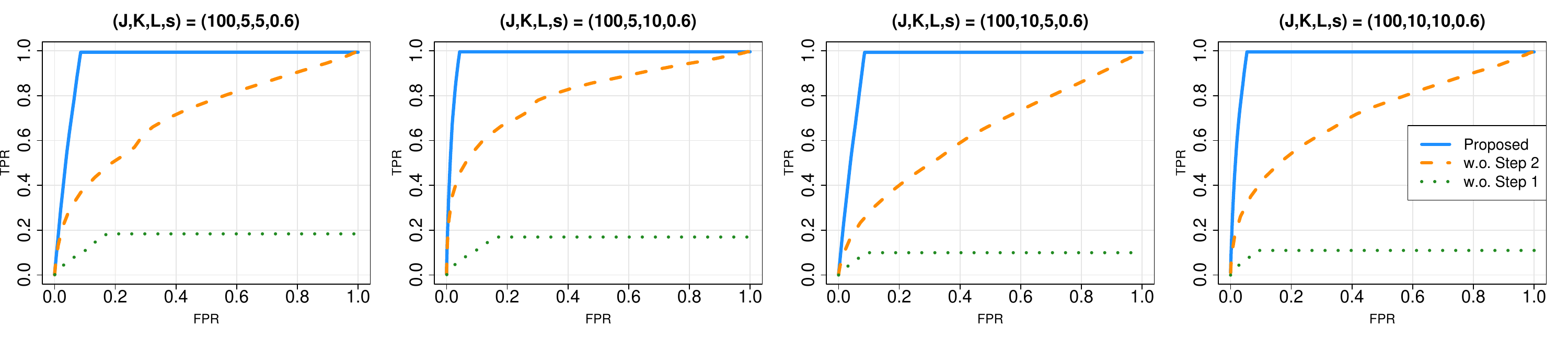}
    \includegraphics[width=\linewidth]{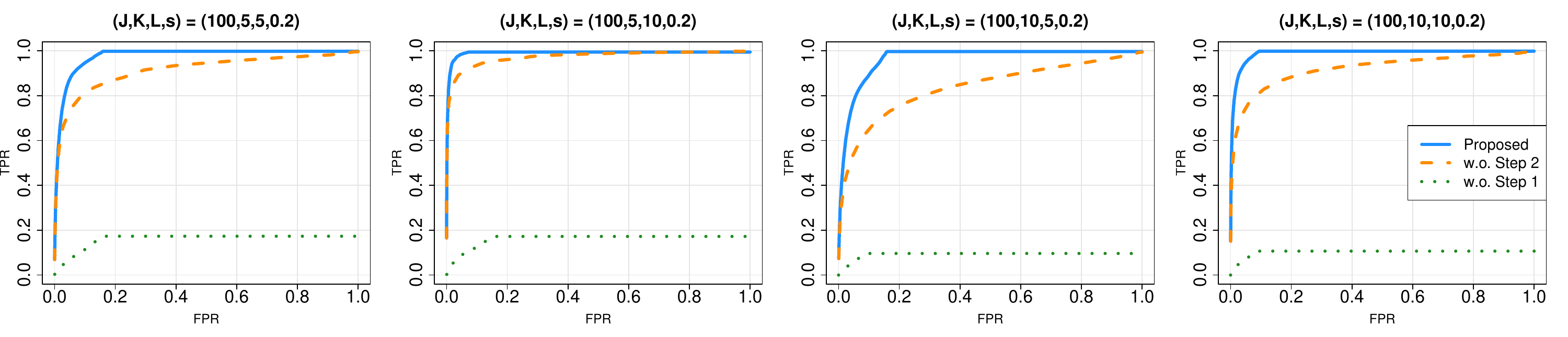}
    \caption{ROC curves (false positive rate versus true positive rate) for selecting $\bOmega$ under $J = 100$ and varying other dimensions/sparsity. The solid line corresponds to the proposed method, whereas the dashed/dotted line corresponds to variants where Steps 2/1 are omitted, respectively.}
    \label{fig:roc}
\end{figure}

We also conduct ablation studies, which reveal that each of Steps 1 and 2 in the proposed method play a crucial role. We report ROC curves regarding precision matrix estimation in \cref{fig:roc}. In particular, without clustering in Step 1, the observed dependence due to the latent class confounding invalidates the block structure estimation in Step 2 (leading to a non-informative ARI of 0.01), and results in a very low TPR. %
When the block structure estimation in Step 2 is omitted, the off-block entries of the precision matrix are subject to spurious false discovery, as indicated by the lower ROC curves.

\section{Real Data Applications}\label{sec:data}
\subsection{ANES public opinion survey}
We analyze the 2022 American National Election Studies (ANES) public opinion pilot survey \citep{anes2022pilot}. After preprocessing as in \cite{chen2024generalized}, the dataset contains ordinal responses from $N = 1511$ participants to $J = 144$ items. The response categories $C_j$ varies from 2 to 7; see \cref{fig:anes_histogram} in the Supplement for a histogram.

Public opinion surveys often violate local independence because the questionnaire design creates dependent item groups. For example, the ANES survey has four consecutive questions that ask respondents whether being White, Black, Hispanic, or Asian comes with advantages (these correspond to nodes 5--8 in \cref{fig:precision_anes}). We show that our method can learn such blocks in a fully data-driven manner. The full data analysis using the proposed pipeline took 16 seconds, with 0.1, 1.8, and 13.7 seconds for Steps 1--3, respectively, including tuning parameter selection in Step 3. We discuss the findings step by step.

\begin{figure}[h!]
    \centering
    \includegraphics[width=\linewidth]{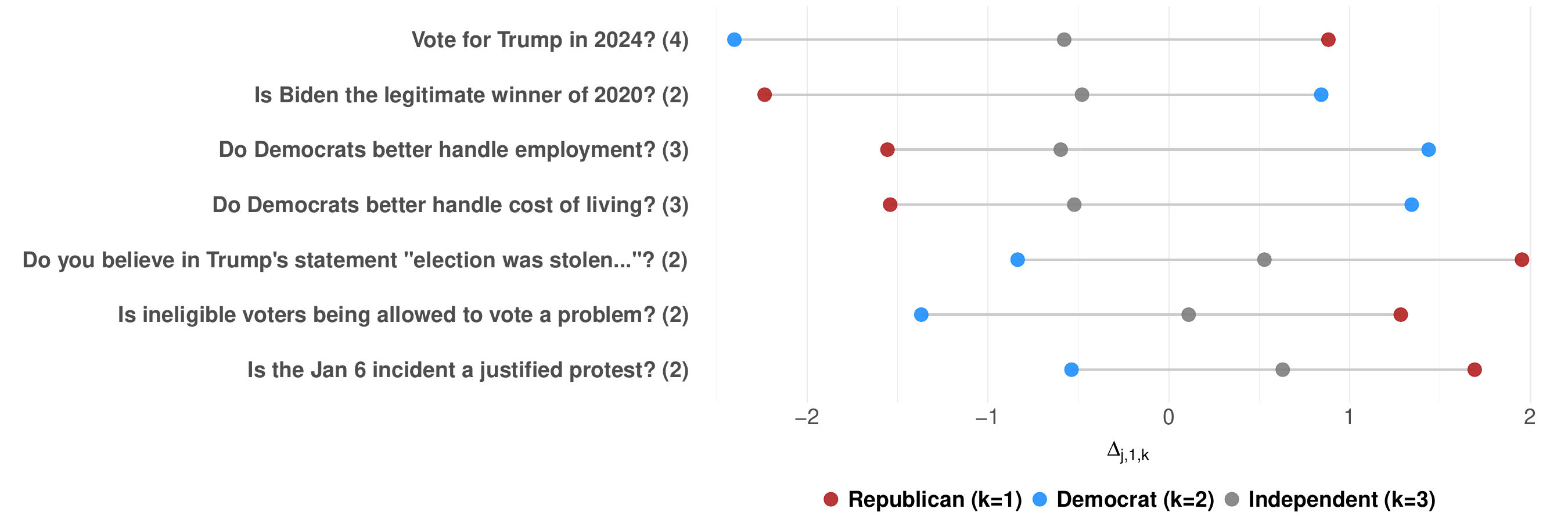}
    \caption{Visualization of the threshold parameter $\Delta_{j,1,k}$ for sample ANES items. Each row corresponds to a sample item (number of categories $C_j$ in parentheses), and each dot corresponds to each latent class $k = 1,2,3$. Note that $\PP(R_{j} = 1 \mid Z = k) = \Phi(\Delta_{j,1,k})$.}
    \label{fig:anes delta}
\end{figure}

\paragraph{Step 1: clustering}
We implemented Step 1 with $K = 3$; see the Supplement for details. By inspecting the estimated threshold parameters $\Delta_{j,1,k}$ for each cluster,
the three clusters can be interpreted as Republican ($k=1$), Democrat ($k=2$), and Independent ($k=3$) voters. Example items are visualized in \cref{fig:anes delta}, where the Republican and Democrat clusters typically form the two extremes.
This interpretation validates the existence of latent classes and shows that spectral clustering can detect them.

\begin{figure}[h!]
    \vspace{-2mm}
    \centering
    \includegraphics[width=\linewidth]{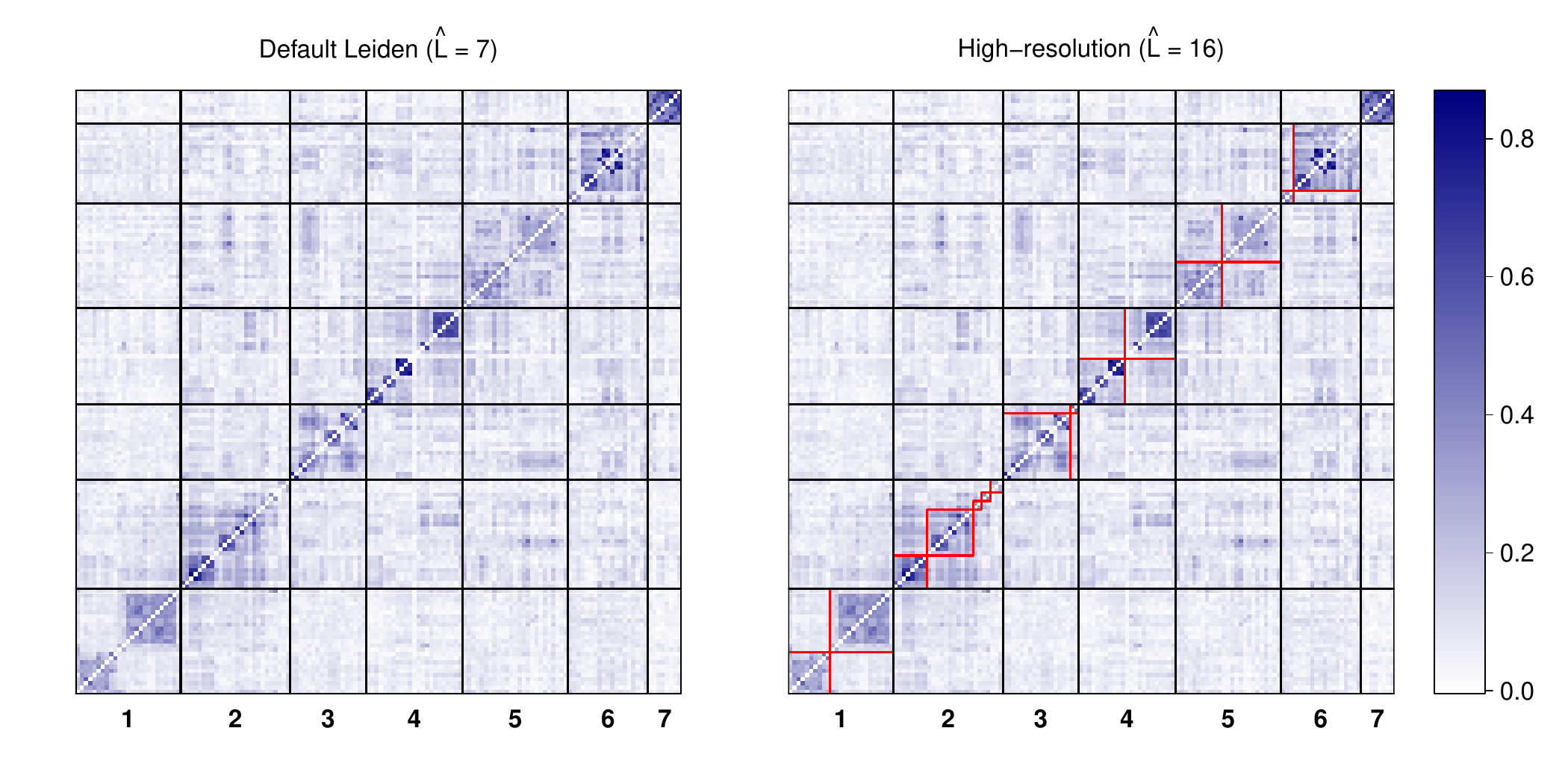}
    \caption{Learned ANES block structures at two resolutions. Both panels show the same thresholded matrix $\tilde{\GG}$ with identical row/column permutations; the left uses default Leiden, and the right uses resolution $1.5$. Red lines indicate the additional fine-grained high-resolution blocks.}
    \label{fig:ANES_block}
    \vspace{-9mm}
\end{figure}

\begin{table}[h!]
    \centering
    \small
        \begin{tabularx}{\textwidth}{@{} l l X @{}}
        \toprule
        \textbf{Default Leiden} & \textbf{High-resolution} & \textbf{Example item (paraphrased)} \\
        \midrule
        \multirow{2}{*}{1. Political Engagement} 
        & Participation & {Joined in a protest march in the past year?} \\
        & Issue ownership & Do Democrats better handle ``employment''? \\
        \midrule
        
        \multirow{6}{*}{2. Trust} 
        & President & Approve of Biden's job as president? \\
        & Experts & Do election officials try to advantage particular candidates? \\
        & Jan 6 incident & Is the Jan 6 incident a justified protest? \\
        & Elections   & Do officials try to get election results right? \\
        & Political tolerance  & Do you favor firing ``gun rights activists''? \\
        \midrule

        \multirow{2}{*}{3. Racism} 
        & Privileges/resentment  & Does being ``Asian'' come with advantages?\\ %
        & Reasons      & Differences due to ``biological reasons''? \\
        \midrule

        \multirow{2}{*}{4. Democratic Attitudes} 
        & Immigration       & Do immigrants make you ``angry''? \\ %
        & Electoral integrity & Is Biden the legitimate winner of 2020? \\
        \midrule

        \multirow{2}{*}{5. Issue Importance} 
        & Economic      & How important is the ``taxes'' issue? \\
        & Social     & How important is the ``gun'' issue? \\
        \midrule

        \multirow{4}{*}{6. Abortion} & Prosecution & Do you favor prosecuting health care professionals who perform abortions? \\
        & Emotions & Does Roe v. Wade being overturned make you feel ``angry''? \\
        \midrule
        
        7. Racial Stereotypes & --- & Are ``White'' people hard-working? \\
        \bottomrule
    \end{tabularx}
  
    \caption{Hierarchical block structure of the ANES dataset across two algorithmic resolutions. The coarse blocks (left) partition into nuanced sub-blocks (middle) under a larger resolution parameter. Example word choices among multiple similar items are indicated in quote marks. The block index $1-7$ in the first column matches that in \cref{fig:ANES_block}.}
    \label{tab:cluster_hierarchy}
    \vspace{-2mm}
\end{table}

\paragraph{Step 2: learning block structures}
We applied our Step 2 with the continuous aggregation matrix $\tilde{\GG}$. \cref{fig:ANES_block} displays the learned block structures at two different algorithmic resolutions. %
In the left panel, we implemented the default Leiden algorithm, and the items are partitioned into $\hat{L} = 7$ blocks. %
In the right panel, we increased the Leiden resolution parameter to $1.5$, which gives $\hat{L} = 16$ blocks. We see that most of the 7 blocks (in black) are partitioned into fine-grained sub-blocks (in red). 
Referring to the questionnaire allowed us to interpret each coarser and finer block (see the first/second columns of \cref{tab:cluster_hierarchy}), revealing an interesting hierarchical structure.
For instance, a coarse block on `issue importance' is split into economic and social issues (see the fifth row in \cref{tab:cluster_hierarchy}). This multi-resolution capability can be practically useful as it allows researchers to adjust the granularity of blocks easily. 

For evaluation, we compute the ARI between $\hat{V}$ and the questionnaire item groups. The questionnaire has $L=31$ groups, including $19$ short blocks (with length four or less); {see \cref{tab:anes true block} in the Supplement}. The default/high-resolution partitions in \cref{fig:ANES_block} achieve ARI of 0.31 and 0.42 
We note that the questionnaire groups are only a guideline rather than ground truth. Indeed, the seemingly mismatched entries are still highly interpretable. For example, the items ``how worried are you about black citizens being the victims of violent crime'' and ``do you think public schools focus too much on race and racism?'' are listed under the respective questionnaire categories of ``crime'' and ``role of schools''. Both items are assigned to the ``racism'' block by our method, revealing nuanced dependence not captured by the default grouping.

\paragraph{Step 3: precision matrix estimation} 
Fixing the block structure with $\hat{L}=7$, we estimate the precision matrices $\bOmega$. Partial-correlation networks ($\rho_{(j,j'),k} := - \omega_{(j,j'),k}/\sqrt{\omega_{(j,j),k} \omega_{(j',j'),k}}$) for selected nodes in the racism block $V_3$ are displayed in \cref{fig:precision_anes}.
For simpler presentation, the nodes are labeled as 1--11; see \cref{tab:anes item matching} in the Supplement for the node information.

\begin{figure}
    \centering
    \includegraphics[width=\linewidth]{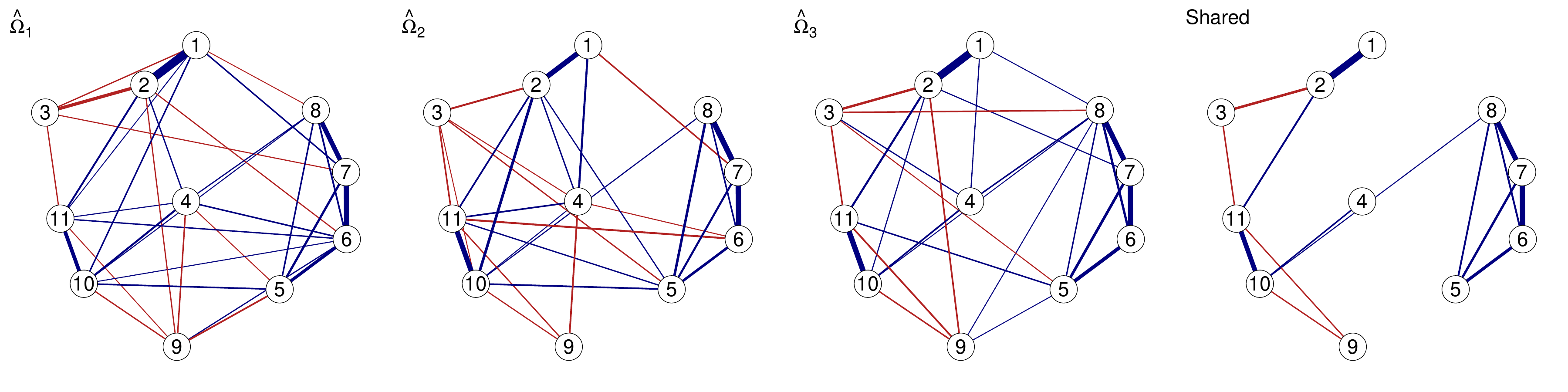}
    \caption{Estimated partial-correlation networks for selected nodes in ANES ``racism'' block. The first three panels correspond to the `Republican', `Democrat', and `Independent' latent classes, and the last panel displays the shared edges. Edge width is proportional to partial correlation, with the blue/red colors denoting positive/negative values. {Details of the 11 nodes are given in \cref{tab:anes item matching} in the Supplement.} %
    } %
    \label{fig:precision_anes}
\end{figure}

The estimated networks show both shared and class-specific dependence. Among the 27, 25, and 22 edges in the first three graphs, 15 are shared. This includes the clique on nodes 5--8 about racial-group advantages for White/Black/Hispanic/Asian people. In contrast, we also see heterogeneous edges across the three classes. A notable discrepancy occurs at the edge 6--11, where the two nodes correspond to the items ``do you think that being Black comes with advantages'' and ``over the past few years, Blacks have gotten less than they deserve''. In the more polarized latent classes (Republicans and Democrats; see the first two panels), these items exhibit conditional dependencies with different signs. However, the third latent class (Independents) exhibits conditional independence. A similar pattern appears for edge 1--7, illustrating how our method captures dependence differences across ideological clusters.

\subsection{HapMap3 genetics data}
We next analyze the HapMap3 dataset, which collects genotypes of individuals from different populations \citep{international2010integrating}.
As the full dataset is extremely high-dimensional, we focus on the $J = 500$ single nucleotide polymorphisms (SNPs) on chromosome 22 with the largest variances. The observed genotype $R_{i,j}$ is encoded as 0, 1, or 2, representing the count of a specific reference allele, so each SNP takes $C_j = 3$ values. We study $N = 479$ individuals from $K = 4$ distinct subpopulations. See Supplement \ref{supp:hapmap discussion} for additional discussion on why block structures are natural here.

Our method took 25 seconds to run, of which Step 1 took 0.1 seconds and Step 2 took 24.6 seconds. Because individual SNP-level dependencies are difficult to interpret biologically, we focus on block-structure recovery and do not report Step 3 for this dataset.

\vspace{-1mm}
\paragraph{Step 1: clustering}
We applied Step 1 with $K = 4$. Compared to the held-out subpopulation information, clustering was nearly perfect, with 11 errors among the $N=479$ individuals.
This illustrates the good performance of the proposed clustering procedure despite local dependence.

\begin{figure}
    \centering
    \includegraphics[width=\linewidth]{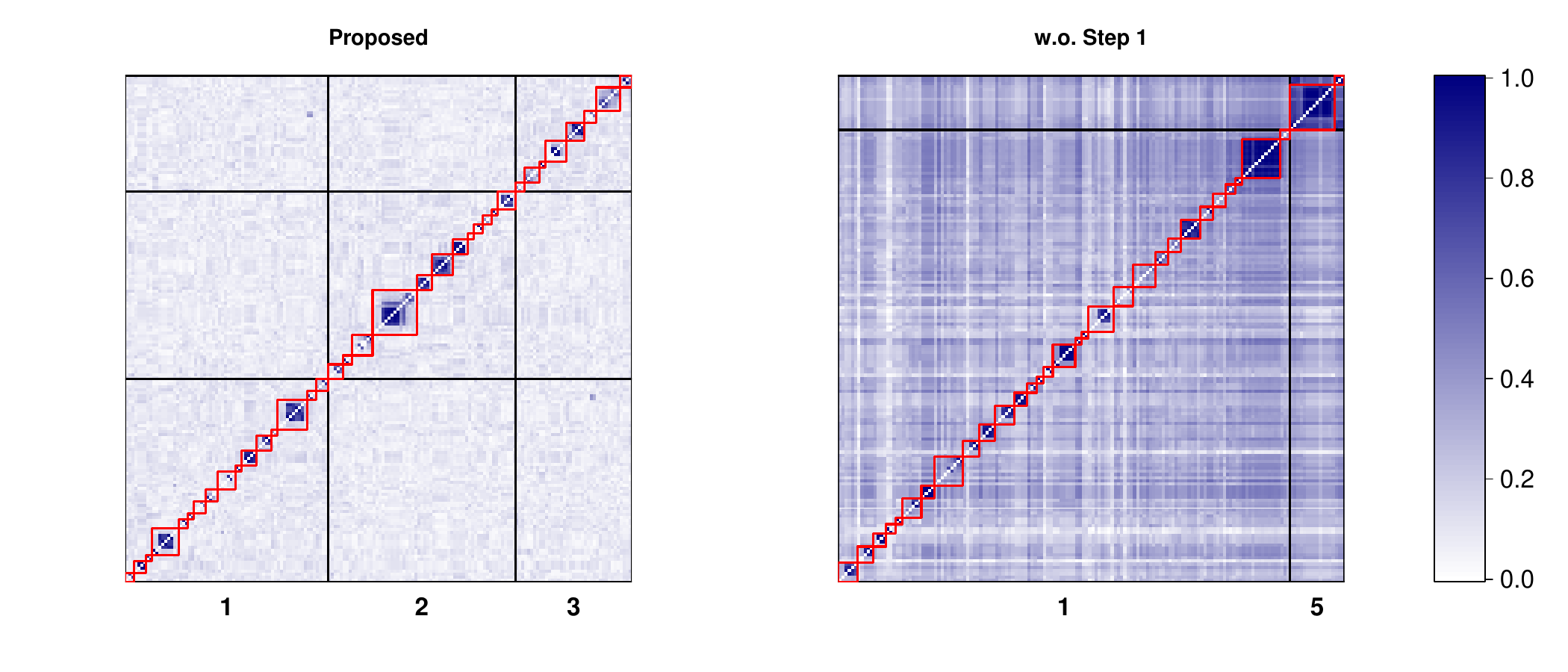}
    \caption{Estimated $\tilde{\GG}$ and community structures for the HapMap3 data. Black/red lines show the default/high-resolution blocks, respectively. Left: the proposed method (Steps 1 and 2). Right: a modified analysis that omits Step 1. The columns of each panel are permuted separately. The proposed analysis leaves much weaker dependence between different learned blocks, visible as lighter off-block regions.}
    \label{fig:hapmap block}
\end{figure}

\vspace{-1mm}
\paragraph{Step 2: learning block structures}
Next, we applied Step 2 to learn block structures. The left panel of \cref{fig:hapmap block} shows the first three large default-resolution blocks, covering 157 SNPs. The default resolution selected $\hat{L}=11$ blocks, while the high-resolution run selected $\hat{L}=76$ finer blocks. The average absolute estimated covariance over off-block SNP pairs is similar under the default and high-resolution partitions: 0.086 and 0.087, respectively. Thus, the additional splits do not increase between-block dependence, supporting the use of the high-resolution partition. The estimated blocks mostly align with the sequential SNP ordering, although some non-adjacent SNPs are assigned to the same block.

Finally, we emphasize that accounting for the population heterogeneity in Step 1 is crucial. When the clustering in Step 1 is omitted, the right panel of \cref{fig:hapmap block} shows much larger cross-block correlations. Quantitatively, the average off-block correlation is $0.287$ without Step 1, compared with the value of $0.087$ after conditioning on the estimated subpopulations.

\section{Discussion}\label{sec:discussion}
{We have proposed an interpretable and theoretically grounded methodology that relaxes local independence in modern high-dimensional latent class analysis. We propose a block-dependent LCM and developed a practical and scalable algorithm that sequentially estimates all model components. The proposed method is supported by rigorous finite-sample theoretical guarantees as well as promising illustrations in both social and biological applications.}

There are several interesting directions for further work. First, our theoretical consistency result builds upon exact recovery of latent classes and block structures. A natural extension is to investigate downstream estimation robustness under weak clustering recovery or block recovery. Methodologically, one may enhance Step 2 of our method by integrating robust polychoric correlation estimators \citep{welz2024robust}, at additional computational cost.

Second, our current framework assumes fully observed responses, whereas many block-dependent ordinal datasets suffer from missingness. For instance, movie ratings may be dependent through shared genres or actors, but each user rates only a small fraction of movies. Hence, extending our methodology to accommodate missingness would be important.

Finally, our unsupervised framework can be naturally extended to supervised settings, such as latent class regression with block-structured dependence, by incorporating covariates to the latent Gaussian mean in \eqref{eq:X given Z}. Also, we can integrate mixed data types (such as continuous) by modifying the polychoric correlation estimator in Step 2 accordingly.

\paragraph{Acknowledgements}
This research is partially supported by NSF Grant DMS-2210796. %

\spacingset{1}
\bibliographystyle{apalike}
\begingroup
\small
\bibliography{ref}
\endgroup

\newpage
\appendix

\appendix
\begin{center}
    \LARGE Supplement to ``Beyond Local Independence: High-Dimensional Latent Class Graphical Models with Shared Block Structure''
\end{center}
\vspace{5mm}

\spacingset{1.6}

The supplement is organized as follows. Section~\ref{supp:literature} discusses related work on locally dependent latent class models, Gaussian graphical and copula models, and other models. Section~\ref{supp:proof} proves all main results from \cref{sec:theory and method}, and Section~\ref{supp:lemmas} proves auxiliary lemmas. Section~\ref{supp:simulation-details} describes implementation details regarding simulations. Section~\ref{supp:comparison-visualization} presents additional comparison studies and visualizations, including a likelihood-based refinement procedure for clustering. Section~\ref{supp:real-data-details} provides additional details and discussion regarding the real data analysis.

\section{Additional Related Work}\label{supp:literature}
\paragraph{Locally dependent latent class models}
There is a steady line of methodological work from quantitative social science and medical diagnosis on relaxing the local independence assumption in LCMs. One classical approach is to directly add pairwise or higher-order interaction effects while modeling the conditional probabilities, assuming a known dependence graph or considering all second-order interactions \citep{harper1972local,hagenaars1988latent,asparouhov2015residual}.
Another classical approach %
is to introduce a unidimensional latent Gaussian random-effect to model dependence in each class \citep{qu1996random, dendukuri2001bayesian,menten2008bayesian}. This line of random-effect modeling has been especially popular in medicine to model conditionally dependent test results.

More related to our work, \cite{everitt1988finite} and \cite{uebersax1999probit} formulated the probit latent class framework (see \eqref{eq:X given Z}--\eqref{eq:R}) for binary and ordinal data. This model has served as a baseline model for developing both Bayesian and frequentist methods that accommodate local dependence \citep{asparouhov2011using,lee2020detecting,lee2022estimating}. See \cite{lee2022estimating} for additional references. In the most general framework, \cite{uebersax1999probit} allowed the threshold vector $\bDelta$ to vary across classes as in \eqref{eq:R}, but later works often considered homogeneous thresholds shared across latent classes. A more important distinction is that all cited works are methodologically driven under a fixed $J$. The Bayesian methods mainly focus on MCMC sampling and the frequentist methods are based on EM algorithms to compute the marginal maximum likelihood estimator, both of which can be concerning when $J$ is large. Our work fills in this gap by proposing a more computationally appealing method with solid theoretical guarantees regarding learning the dependence structure for high-dimensional data with increasing $J$.

\paragraph{Gaussian graphical models, copulas, and network models}
Gaussian graphical models and copulas (also called nonparanormals) have been studied extensively to model dependent continuous data \citep{yuan2007model,liu2009nonparanormal,liu2012high}. %
To allow various response types, extensions have been proposed for binary data \citep{fan2017high} as well as ordinal/mixed data \citep{guo2015graphical,feng2019high,quan2018rank} and truncated data \citep{yoon2020sparse}. %
However, the vast majority of these works do not consider \emph{heterogeneous latent sources}, as considered in this paper. Note that ignoring potential latent classes and fitting a single graphical model may lead to inconsistent estimation, and it is important to model such heterogeneous populations. Also, most methods (excluding \cite{guo2015graphical}, which proposes a variational EM algorithm) rely on Kendall's tau as discussed in \cref{rmk:rank}. A few exceptions include \cite{Chandrasekaran2012latent,frot2019robust,frot2019graphical}, which take a different approach from ours by using linear structural equations to model the relationship between continuous latent variables and observations. 

We also note another line of work that assumes knowledge of the heterogeneous sources (that is, with fully observed class labels) and posits separate Gaussian graphical models for each source \citep{guo2011joint,danaher2014joint,lee2015joint,ma2016joint}. %
In particular, \cite{guo2011joint} considers a related notion of shared sparsity across sources, but without block structures. \cite{ma2016joint} assumed prior knowledge of the structural relationship between multiple graphical models.
\cite{danaher2014joint}, \cite{mohan2014node}, and \cite{yang2015fused} used a shared block structure among the multiple precision matrix estimates to speed up computation. %
Unlike this paper, the block structure in the above studies is not an interpretable modeling assumption to capture unobserved population heterogeneity, but is a structural artifact that arises from solving a strongly penalized optimization problem, so the estimated structure crucially depends on the penalty parameter. 
Extending such Gaussian likelihood-based methods to ordinal data is non-trivial, and we instead estimate the block structure directly from the covariance matrices (instead of simultaneously learning the precision matrices). This way, we can circumvent introducing an additional joint penalty regarding the similarity of precision matrices. More recent works also considered multiple graphical models using Gaussian copulas \citep{liu2023decomposition,hermes2024copula}.

Additionally, continuous latent variable models based on non-Gaussian copulas have been proposed in the item response theory literature in psychometrics \citep{braeken2007copula,braeken2011boundary}. These works also consider similar block structures within the items, but assume that the block structure is known. Focusing exclusively on binary data, latent variable network models based on non-Gaussian graphical modeling (such as Ising models) have also been introduced \citep{chen2018robust,miao2024bayesian}, and have been recently extended to ordinal data as well \citep{marsman2025bayesian}.
However, these works assume the network dependence to be invariant with respect to the latent variables, and the estimation methods do not come with theoretical guarantees.

\paragraph{Directed graphical models}
This paper generalizes the usual latent class model by allowing \emph{undirected} dependence structures among the observations, and we compare it with work in mixture modeling and causal discovery that considers \emph{directed} dependence. While this problem has been popular for continuous data with linear structural equations \citep{Chandrasekaran2012latent,xie2024generalized}, discrete data has received less attention. Recently, \cite{gordon2023causal}, \cite{mazaheri2023causal}, and \cite{lee2026deep} considered a directed analog of the problem studied here, where the first two focus mainly on binary responses, and the third allows more complex latent structures but under a generalized linear parametrization.
This line of work treats the causal graph among the observed variables as fixed across latent classes, whereas our framework allows the within-block graph to vary. Also, their theoretical framework focuses on asymptotic population identifiability, whereas we establish non-asymptotic error bounds and allow high-dimensional responses.

\paragraph{Multi-layer networks}
Our model can also be viewed through the lens of multi-layer networks: each latent class induces a weighted graphical model with precision matrix $\bOmega_k$ on the same items $V = [J]$. The within-block edges may vary across latent classes, but the target block partition is shared across layers. This connects our block-recovery step to work on planted partitions in multi-layer networks, which has been of both theoretical interest (in the context of multiple stochastic block models) %
\citep{paul2020spectral,ma2023community,chatterjee2024detecting} and applied interest (for brain and social networks) \citep{bassett2011dynamic,greene2013producing}. One distinction is that those works typically observe binary adjacency matrices, whereas here the weighted layer-specific graphs are estimated from tabular ordinal responses.

\section{Proof of Main Results}\label{supp:proof}
\subsection{Proof of results from \cref{subsec:spectral}}

\begin{proof}[Proof of \cref{prop:clustering}]
    For simplicity, set $\tilde{J} := \sum_{j=1}^J C_j$. 
    We equivalently write \eqref{eq:conditional expectation extended} as
    $$\tilde{R}_{i,\sum_{m=0}^{j-1} C_m + a} = \theta_{j,a,k} + \epsilon_{i,\sum_{m=0}^{j-1} C_m + a},$$
    where $\bepsilon_i = (\epsilon_{i,1},\ldots,\epsilon_{i,\tilde{J}})^\top \in [-1,1]^{\tilde{J}}$ is a mean-zero noise vector that is independent across $i$. Under the block structure assumption, we can decompose $\bepsilon_i$ into $L$ independent components $\big(\bepsilon_i^{(1)},\ldots,\bepsilon_i^{(L)}\big)$. Then, denoting $\|\cdot\|_{\psi_2}$ as the sub-Gaussian norm of a random vector (see Definition 2.5.6, \cite{vershynin2018high}), we have 
    $$\|\bepsilon_i\|_{\psi_2} = \max_{\ell=1}^L \|\bepsilon_i^{(\ell)}\|_{\psi_2} \le \sqrt{\tilde{J}_{\max}},$$
    so the mean-zero noise $\bepsilon_i$ is sub-Gaussian with norm bounded by $\sqrt{\tilde{J}_{\max}}$. Here, the last inequality follows as $\bepsilon_i^{(\ell)}$ is a bounded random vector with dimension $\sum_{j \in V_\ell} C_j \le \tilde{J}_{\max}$. %

    Given the sub-Gaussian bound for the noise, we show \eqref{eq:clustering accuracy} by invoking \cite{zhang2024leave}. Set the minimum distance between the cluster centers as
    $$D_{\theta} := \min_{k \neq k' \in [K]} \|\btheta_{k} - \btheta_{k'}\|,$$
    where $\btheta_k = (\theta_{1,1,k}, \ldots, \theta_{J,C_J,k})^\top$ denotes the flattened mean vector of the $k$th cluster.
    We first show that the associated signal-to-noise ratios in \cite{zhang2024leave} are large enough:
    \begin{align}\label{eq:spectral clustering snr}
        \psi_1 := \frac{D_\theta}{K(1+\sqrt{J/N})\sqrt{\tilde{J}_{\max}}} > c, \quad \rho_1 := \frac{\sigma_K\big(\EE[\tilde{\RR} \mid \ZZ]\big)}{(\sqrt{N}+\sqrt{J})\sqrt{\tilde{J}_{\max}}} >c.
    \end{align}
    Below, we separately show each claim of \eqref{eq:spectral clustering snr} under the assumption
    $$D \gtrsim K \sqrt{\tilde{J}_{\max}} \Big(1+\sqrt{\frac{J}{N}}\Big), \quad \sigma_K(\bTheta) \gtrsim \sqrt{K \tilde{J}_{\max}} \Big(1+\sqrt{\frac{J}{N}}\Big).$$
    
    For showing the second bound in \eqref{eq:spectral clustering snr}, we work under a high probability set $E_1 := \{\min_{k\in[K]} N_k > \frac{\delta_3 N}{2K}\}$, where $N_k := \sum_{i=1}^N \mathbb{I}(Z_i=k)$ denotes the size of the $k$th latent class and $\delta_3$ is the constant from \cref{assmp:bounded parameter}. By the multiplicative Chernoff bound for the sum of bounded random variables, we get
    $$\PP\left(N_k\le \frac{N \pi_k}{2}\right) = \PP\left(N_k\le \frac{\EE[N_k]}{2}\right) \le \exp\left(-\frac{N\pi_k}{8}\right).$$
    Using $\pi_{k}\ge \delta_3/K$ from \cref{assmp:bounded parameter} and a union bound, we have
    \begin{equation}\label{eq:e_1}
        \PP(E_1^c) \le \sum_{k=1}^K \PP\left(N_k\le \frac{N\pi_k}{2}\right) \le K\exp\left(-\frac{\delta_3 N}{8K}\right).
    \end{equation}
    As $\mathbf{Y}$ is the $N \times K$ one-hot membership matrix, we have $\mathbf{Y}^\top \mathbf{Y}=\operatorname{diag}(N_1,\ldots,N_K)$. Hence, on $E_1$, $$\sigma_K(\mathbf{Y})=\sqrt{\min_{k\in[K]}N_k} \ge \sqrt{\delta_3 N / 2K}.$$
    Recalling the decomposition $\EE[\tilde{\RR} \mid \ZZ] = \mathbf{Y} \bTheta^\top,$
    we have
    $$\sigma_K\big(\EE[\tilde{\RR} \mid \ZZ]\big) \ge \sigma_K(\mathbf{Y}) \sigma_K(\bTheta) \gtrsim (\sqrt{N}+\sqrt{J}) \sqrt{\tilde{J}_{\max}}.$$
    Here, the first inequality is a standard fact regarding singular values, and the second inequality follows from invoking the assumption on $\sigma_K(\bTheta)$.

    Next, to show the first bound in \eqref{eq:spectral clustering snr}, it suffices to prove that $D_\theta \gtrsim D = \min_{k \neq k' \in [K]} \|\bDelta_{k} - \bDelta_{k'}\|$, since the separation in the proposition is defined in terms of the threshold vectors $\bDelta_k$ as opposed to the cluster-center vectors $\btheta_k$.
    Fix any $k \neq k' \in [K]$ and define
    $$d_{j,a} = d_{j,a,k,k'} := \Phi(\Delta_{j,a,k}) - \Phi(\Delta_{j,a,k'}), \quad \forall a \in [C_j-1],$$
    and $\mathbf d_j = (d_{j,1}, \ldots, d_{j,C_j-1})^\top$. For notational convenience, also let $d_{j,0} = d_{j,C_j} = 0$ so that
    $$\theta_{j,a,k} - \theta_{j,a,k'} = d_{j,a} - d_{j,a-1}.$$
    Hence, we have
    \begin{align*}
        \|\btheta_{j,k}-\btheta_{j,k'}\|^2 = \sum_{a=1}^{C_j} (d_{j,a} - d_{j,a-1})^2 = \mathbf d_j^\top \mathbf A \mathbf d_j,
    \end{align*}
    where $\mathbf A$ is a $(C_{j}-1) \times (C_j-1)$ tri-diagonal matrix with $2$ on the main diagonal and $-1$ on the off-diagonals. Noting that $\lambda_{\min}(\mathbf A) = 2 - 2 \cos({\pi}/{C_j}) > 0$ (see \cite{yueh2005eigenvalues}), we can write
    $$\|\btheta_{j,k}-\btheta_{j,k'}\|^2 \gtrsim \|\mathbf d_{j}\|^2 \gtrsim \|\bDelta_{j,k} - \bDelta_{j,k'}\|^2,$$
    so the claim holds. Here, the hidden constant in the $\gtrsim$ symbols depend on $C_{\max}$ and $M$, respectively.

    Finally, we complete the proof. By intersecting on $E_1$ and using Theorem 3.1 from \cite{zhang2024leave} (note that their loss function is with a $1/N$ scaling, also the condition $N \gtrsim K^2$ is required for this result), we have
    \begin{align*}
        \EE [\ell (\hat{\ZZ}, \ZZ) \mid \ZZ] &= \EE[ \ell (\hat{\ZZ}, \ZZ) \mathbb{I}(\ZZ \in E_1) \mid \ZZ] + \EE [\ell (\hat{\ZZ}, \ZZ) \mathbb{I}(\ZZ \in E_1^c) \mid \ZZ] \\
        &\le N \Big(e^{-(1-c_2/\psi_1 - c_2/\rho_1^2) \frac{D_\theta^2}{8 \tilde{J}_{\max}}} +e^{-\frac{N}{2}}\Big) + N \mathbb{I}(\ZZ \in E_1^c) \\
        &\le N \Big(e^{- c_3 \frac{D^2}{\tilde{J}_{\max}}} +e^{-\frac{N}{2}}\Big) + N \mathbb{I}(\ZZ \in E_1^c).
    \end{align*}
    Here, the last line used \eqref{eq:spectral clustering snr} and $D_\theta \gtrsim D$. Now, \eqref{eq:clustering accuracy} follows from taking an outer expectation over $\ZZ$, and using the bound \eqref{eq:e_1}. Here, note that the middle term is dominated by \eqref{eq:e_1}. 
\end{proof}

\begin{proof}[Proof of \cref{prop:pi delta consistency}]
    We first show the bound for $\bpi$. By intersecting with the event of exact clustering $E_2 := \{\hat{\ZZ} = \ZZ\}$ (for simplicity, assume that the labels are correctly permuted), we have 
    \begin{align*}
        \PP(\max_{k \in [K]}|\hat{\pi}_k - \pi_k| \ge t) &\le   \PP\Big(\max_{k \in [K]}\Big|\frac{1}{N} \sumin \mathbb{I}(Z_i = k) - \pi_k\Big| \ge t\Big) + p_1 \\
        &\le \sum_{k=1}^K \PP\Big(\Big|\frac{1}{N} \sumin \mathbb{I}(Z_i = k) - \pi_k\Big| \ge t\Big) + p_1  \le 2 K e^{-2 N t^2} + p_1.
    \end{align*}
    For the first inequality, we used \cref{prop:clustering} to get $\PP(E_2^c) \le p_1$ ($p_1$ is as in \eqref{eq:clustering accuracy}). The second inequality follows from a union bound, and the final inequality follows from Hoeffding's inequality (see Theorem 2.2.6 in \cite{vershynin2018high}).

    Next, we show the bound for $\bDelta$. We work under $E_1$ and $E_2$, where $E_1$ is as in the proof of \cref{prop:clustering}. %
    For any $t < M$, by intersecting with $E_1,E_2$, we have
    \begin{align*}
         \PP(\sup_{j,a,k} &|\hat{\Delta}_{j,a,k} - \Delta_{j,a,k}| > t) \le \sum_{j,a,k} \PP( |\hat{\Delta}_{j,a,k} - \Delta_{j,a,k}| > t, E_1, E_2) + \PP(E_1^c) + \PP(E_2^c) \\
        &\le \sum_{j,a,k} \PP\Big(\Big|\frac{1}{\hat{N}_k} \sum_{i: \hat{Z}_i=k} \mathbb{I}(R_{i,j} \le a) - \Phi(\Delta_{j,a,k}) \Big| > \frac{t}{M'}, E_1, E_2\Big) + K e^{-c N/K} + p_1 \\
        &\le \sum_{j,a,k} \PP\Big(\Big|\frac{1}{N_k} \sum_{i: Z_i=k} \mathbb{I}(R_{i,j} \le a) - \Phi(\Delta_{j,a,k}) \Big| > \frac{t}{M'}, E_1 \Big) + p_1 \\
        &\le \sum_{j,a,k} 2 \EE \Big[\exp \Big(- \frac{2 N_k t^2}{M'^2}\Big) \mathbb{I}(E_1) \Big]+ p_1 \\
        &\le 2 \Big(\sum_{j=1}^J C_j \Big) K e^{- \frac{2c t^2}{M'^2} \frac{N}{K}} + p_1.
    \end{align*}
    Here, the first inequality follows from a union bound. The second line uses (i) the fact that the function $\Phi^{-1}$ is $M'$-Lipschitz with $M' = 1/\phi(2M)$ under the domain $[\Phi(-2M), \Phi(2M)]$ (cf. Lemma A.1 in \cite{fan2017high}), (ii) the bounds for $E_1, E_2$ from \cref{prop:clustering}.
    The third line uses the definition of $E_2$ alongside $p_1 \gtrsim K e^{-cN/K}$, and the fourth line follows from the tower property and Hoeffding's inequality (conditioned on $\ZZ$). The final line follows from using the definition of $E_1$ to lower bound each $N_k$.  This gives the desired bound.
\end{proof}

\subsection{Proof of results from \cref{subsec:covariance matrix}}
We prove \cref{prop:covariance estimation} using the following lemma regarding finite-sample error bounds for pairwise polychoric correlations. In the statement of the lemma, we use a simplified notation by setting $j = 1, j'= 2$ and omitting the latent class index $k$.
\begin{lemma}\label{lem:polychoric}
    Assume $\begin{pmatrix}
        X_1^{(\rho)} \\ X_2^{(\rho)}
    \end{pmatrix} \sim N\left(\begin{pmatrix}
        0 \\ 0
    \end{pmatrix}, \begin{pmatrix}
        1 & \rho \\ \rho & 1
    \end{pmatrix} \right)$, where $|\rho|\le 1-\delta_1$. For $a \in [C_1], b \in [C_2]$, let
    $$\pi_{a,b}(\rho; \bDelta_1, \bDelta_2) := \PP\big(X_1^{(\rho)} \in [\Delta_{1,a-1}, \Delta_{1,a}), X_2^{(\rho)} \in [\Delta_{2,b-1}, \Delta_{2,b})\big).$$
    Define the threshold-vector estimate $\hat{\bDelta}$ as
    $$\hat{\Delta}_{j,a} := \Phi^{-1} \left( \frac{1}{{N}} \sum_{i=1}^N \mathbb{I}(R_{i,j} \le a) \right), \quad \forall a \in [C_j-1], \quad j = 1,2.$$
    For population parameters $\rho^*,\bDelta_1^*,\bDelta_2^*$, define the sample and population loss functions as
    \begin{align*}
        \ell_N(\rho) &:= \frac{1}{N} \sumin \log \pi_{R_{i,1}, R_{i,2}}(\rho; \hat{\bDelta}), \\
        \ell_\infty(\rho) &:= \EE \log \pi_{R_{i,1}, R_{i,2}}(\rho; \bDelta^*) = \sum_{a \in [C_1], b \in [C_2]} \pi_{a,b}(\rho^*; \bDelta^*) \log \pi_{a,b}(\rho; \bDelta^*).
    \end{align*}
    Then, for $\hat{\rho}_N := \argmax \ell_N(\rho)$, there exist absolute constants $c_1, c_2, c_3 > 0$ (that only depend on $C_1, C_2,\delta_1,\delta_2,M$) such that
    $$\PP(|\hat{\rho}_N - \rho^*| > t) \le c_1 (e^{-c_2 N t^2} + e^{-c_3 N}), \quad \forall t>0.$$
\end{lemma}

\begin{proof}[Proof of \cref{prop:covariance estimation}]
    We first prove the tail bound \eqref{eq:tail sigma}. Set $E_1, E_2$ as before, so that $\PP(E_1^c) \le K e^{-cN/K}$. On $E_2$, the estimator $\hat{\bSigma}_k$ is exactly the oracle estimator $\hat{\bSigma}_k^*$ computed using the true labels. Conditional on the labels, the oracle estimator for class $k$ is based on $N_k$ observations; hence, by applying \cref{lem:polychoric} with $N=N_k$, we get
    $$\PP(\|\hat{\bSigma}_k^* - \bSigma_k\|_{\max} > t, E_1) \le c_1 J^2 \Big(e^{-c_2 c N t^2/K} + e^{-c_3 c N/K}\Big)$$
    for some constants $c_1, c_2, c_3 > 0$.

    Hence, we can write 
    \begin{align*}
        \PP\Big(\max_{k=1}^K \|\hat{\bSigma}_k - \bSigma_k\|_{\max} > t\Big) &\le \PP\Big(\max_{k=1}^K \|\hat{\bSigma}_k^* - \bSigma_k\|_{\max} > t\Big) + \PP(E_2^c) \\
        &\le \sum_{k=1}^K \PP\Big(\|\hat{\bSigma}_k^* - \bSigma_k\|_{\max} > t, E_1\Big) + \PP(E_1^c) + \PP(E_2^c) \\
        &\le c_1 K J^2 \Big(e^{-c_2 N t^2/K} + e^{-c_3 N/K}\Big) + p_1,
    \end{align*}
    where the second line follows from a union bound, and $p_1$ dominates $\PP(E_1^c)$ in the last line. By taking $t := c \sqrt{\frac{K\log(J K)}{N}}$ with a large enough constant $c$ and using the assumptions $N/K \gtrsim \log(JK), p_1 = O(N^{-1})$, the RHS becomes $O\left((JK)^{-1} + N^{-1}\right)$.

\end{proof}

\begin{proof}[Proof of \cref{prop:block structure}]
    It suffices to show that the aggregated graph $\GG$ (defined by the binary matrix $\GG$) has no edges between distinct true blocks and contains the within-block connected components $\cup_{\ell}E_\ell$ with high probability. To see why, recall that \cref{assump:minimal strength} assumed every true block $V_\ell$ to be connected. The absence of cross-block edges implies that no connected component of $\GG$ can contain vertices from two different true blocks. Therefore, the connected components of $\GG$ are exactly $V_1,\ldots,V_L$.

    To show the claim, work on the high probability set $F := \left\{\max_{k=1}^K \|\hat{\bSigma}_k-\bSigma_k\|_{\max}\le \tau_1\right\}.$
    If $j$ and $j'$ belong to different true blocks, then $\sigma_{(j,j'),k}=0$ for all $k\in[K]$ by the block-diagonal structure of $\bSigma_k$. On $F$, this implies $|\hat{\sigma}_{(j,j'),k}|\le \tau_1$ for all $k$, and hence $g_{j,j'}=0$.     
    On the other hand, if $\max_{k\in[K]}|\sigma_{(j,j'),k}|>\epsilon_N \ge 2 \tau_1$, then there exists some $k\in[K]$ such that 
    $$|\hat{\sigma}_{(j,j'),k}|\ge |\sigma_{(j,j'),k}|-\tau_1 > \tau_1.$$
    Therefore, $\max_{k\in[K]}|\hat{\sigma}_{(j,j'),k}| > \tau_1$ and $g_{j,j'}=1$. This proves the desired recovery of off-block zeros and within-block connected components.
\end{proof}

\subsection{Proof of results from \cref{subsec:sparsity}}
\begin{proof}[Proof of \cref{prop:precision matrix}]
\begin{enumerate}[(a)]
    \item We work under the high probability event of correct block recovery given in \cref{prop:block structure}. The conclusion follows from equation (13) in \citep[Theorem 6,][]{cai2011constrained}. Here, we take $\rho = 0$ and note that no parameter space assumptions are required for eq. (13) therein.
    \item This is immediate by part (a) and the additional assumption that $\epsilon_N' \ge 8M_N' \lambda$.
\end{enumerate}
\end{proof}

\section{Proof of Lemmas}\label{supp:lemmas}
To prove \cref{lem:polychoric}, we use the following two lemmas. 
The first lemma is a restatement of the usual MLE consistency argument (see Theorem 5.7 in \cite{van2000asymptotic}).
    \begin{lemma}\label{lem:consistency}
        Fix any $\epsilon>0$. Let $\ell_N, \ell_\infty$ be functions and suppose there exists $\eta = \eta(\epsilon)$ such that
        $$\PP\Big(\sup_\rho |\ell_N(\rho)-\ell_\infty(\rho)| > \eta\Big) \le p(\eta), \quad \ell_\infty(\rho^*) > \ell_\infty(\rho) + 2 \eta$$
        for all $\rho$ with $ |\rho-\rho^*|\ge \epsilon$.
        Then, we have $|\hat{\rho} - \rho^*| \le \epsilon$ with probability $1-p(\eta)$.
    \end{lemma}

The second lemma is a technical result to ensure that the proportion parameter $\pi_{a,b}$ and its derivative are well-defined and bounded. Its proof is postponed to the end of the section.

\begin{lemma}\label{lem:lower bound of pi}
    Assume that $|\rho| \le 1-\delta_1$, $\max_{a,b} (|\Delta_{1,a}|,|\Delta_{2,b}|) \le M$, and $\min_{a,b}(|\Delta_{1,a} - \Delta_{1,a-1}|,|\Delta_{2,b} - \Delta_{2,b-1}|) \ge \delta_2$. Here, $a \in [C_1-1],b\in [C_2-1]$.
    \begin{enumerate}[(a)]
        \item For some constant $c = c(\delta_1,\delta_2,M) > 0$, we have $\pi_{a,b}(\rho;\bDelta) > c$.
        \item For some constant $L = L(\delta_1,\delta_2,M) > 0$, we have $0<\frac{\partial \log \pi_{a,b}(\rho;\bDelta)}{\partial \Delta_{1,a}} < L$. Similar bounds for other derivatives also hold.
    \end{enumerate}
\end{lemma}

\begin{proof}[Proof of \cref{lem:polychoric}]
    Throughout the proof, the lowercase $c$ will denote absolute constants whose exact values may vary, and we shall omit the dependence on the values $C_1,C_2,\delta_1,\delta_2,M$. Following the notation convention, let $\bDelta = (\bDelta_1^\top, \bDelta_2^\top)^\top$. Noting that $\pi_{a,b}$ is bounded away from zero (see \cref{lem:lower bound of pi}), both $\partial_\rho^3 \ell$ and $\partial_\rho^2 \partial_{\bDelta} \ell$ are uniformly bounded above by some constants $M_3 = M_3(\delta_1,\delta_2,M)$ and $M_{\bDelta} = M_{\bDelta}(\delta_1,\delta_2,M)$. Also, by direct computation, we get
    \begin{align}\label{eq:kappa def}
        -\ell_\infty''(\rho^*) = \sum_{a,b} \frac{\partial_\rho \pi_{a,b}(\rho^*; \bDelta_1^*, \bDelta_2^*)^2}{\pi_{a,b}(\rho^*; \bDelta_1^*, \bDelta_2^*)} \ge \frac{\phi_2(\Delta_{1,1}^*, \Delta_{2,1}^*; \rho^*)^2}{\pi_{1,1}(\rho^*; \bDelta_1^*, \bDelta_2^*)} := \kappa(\rho^*) > 0.
    \end{align}
    The inequality uses the summand corresponding to the index $a=b=1$, and $\kappa(\rho^*)$ is a positive constant that depends on the true parameters. Note that $\kappa(\rho^*)$ can be further lower bounded by $c(\delta_1,M)$ by a similar bound as in \cref{lem:lower bound of pi}(a).
    
    \paragraph{Step 0. Preliminary error bounds.}
    In the remainder of the proof, we will work on the event
    \begin{align}\label{eq:event_delta}
            E_1 &:= \left\{\|\hat{\bDelta}-\bDelta\|_\infty \le \frac{\kappa(\rho^*)}{6 M_{\bDelta}}\right\}, \quad E_2 := \left\{|\hat{\rho}-\rho^*| \le \frac{\kappa(\rho^*)}{6 M_3}\right\}.
    \end{align}
    We start by claiming that both events occur with exponentially high probability, that is, $\PP(E_1), \PP(E_2) \ge 1 - O(e^{-c N}).$
    
    We first show the claim for $E_1$. This follows by applying the second conclusion in \cref{prop:pi delta consistency} with the choices $K=1,J=2$ and $t=\frac{\kappa(\rho^*)}{6 M_{\bDelta}}$ as in \eqref{eq:event_delta}. Note that this choice of $t$ is a constant, and is absorbed into the generic constant $c$.

    Next, we move on to bounding $E_2$. 
    By applying \cref{lem:consistency} with the choices $\epsilon = \frac{\kappa(\rho^*)}{6 M_3}$ and $\eta = \eta(\epsilon)$, it suffices to show the uniform law with $p(\eta) = e^{-c(\eta) N}$. For this goal, define
    $$\ell_N^*(\rho) := \frac{1}{N} \sumin \log \pi_{R_{i,1},R_{i,2}} (\rho; \bDelta^*)$$
    and write
    \begin{align}
        \PP\Big(\sup_\rho |\ell_N(\rho)-\ell_\infty(\rho)| \ge \eta \Big) &\le \PP\Big(\sup_\rho |\ell_N(\rho)-\ell_N^*(\rho)| \ge \frac{\eta}{2} \Big) + \PP\Big(\sup_\rho |\ell_N^*(\rho)-\ell_\infty(\rho)| \ge \frac{\eta}{2} \Big) \notag \\
        &\le \PP\Big(\sup_\rho \sup_{a,b} |\log \pi_{a,b}(\rho; \hat{\bDelta}) - \log \pi_{a,b}(\rho; \bDelta^*)| \ge \frac{\eta}{2} \Big) + e^{-c N \eta^2} \label{eq:second_line}\\
        &\le \PP\Big(\|\hat{\bDelta} - \bDelta^*\| \ge \frac{\eta}{2L} \Big) + e^{-c N \eta^2} \lesssim e^{-c N \eta^2}. \notag
    \end{align}
    Here, the first line uses a triangle inequality. The second line follows by noting that
    \begin{align*}
        |\ell_N(\rho)-\ell_N^*(\rho)| &= \frac{1}{N} \Big|\sumin \sum_{a,b} \mathbb{I}(R_{i,1}=a, R_{i,2}=b) \Big(\log \pi_{a,b}(\rho;\hat{\bDelta}) - \log \pi_{a,b}(\rho;\bDelta^*)\Big) \Big| \\
        &\le \sup_{a,b} \Big|\log \pi_{a,b}(\rho;\hat{\bDelta}) - \log \pi_{a,b}(\rho;\bDelta^*)\Big|
    \end{align*}
    (for the first term in \eqref{eq:second_line}), and the uniform version of Hoeffding's inequality (for the second term in \eqref{eq:second_line}). The third line follows from noting that $\log \pi_{a,b}(\rho; \cdot)$ is $L$-Lipschitz where $L$ is uniform in $a,b,\rho$ (see part (b) of \cref{lem:lower bound of pi}), followed by the concentration result for $\hat{\bDelta}$ in \cref{prop:pi delta consistency}. %

    \paragraph{Step 1. Score equation.}
    Define the score function (the derivative of $\ell_N$) as $$S_N(\rho, \hat{\bDelta}) := \frac{1}{N} \sumin \partial_\rho \log \pi_{R_{i,1}, R_{i,2}}(\rho; \hat{\bDelta}_1, \hat{\bDelta}_2).$$
    By definition, we have $S_N(\hat{\rho}, \hat{\bDelta}) = 0$. By a Taylor expansion, we can write
    $$0 = S_N(\rho^*, \bDelta^*) + \partial_\rho S_N(\tilde{\rho}, \tilde{\bDelta}) (\hat{\rho} - \rho^*) + \partial_{\bDelta} S_N(\tilde{\rho}, \tilde{\bDelta})^\top (\hat{\bDelta} - \bDelta^*)$$
    for some $(\tilde{\rho}, \tilde{\bDelta}) \in (\hat{\rho}, \rho^*) \times (\hat{\bDelta}, \bDelta^*).$
    Rearranging the terms, we get
    \begin{align}\label{eq:score equation}
        \hat{\rho}-\rho^* = -[\partial_\rho S_N(\tilde{\rho}, \tilde{\bDelta})]^{-1} \big( S_N(\rho^*, \bDelta^*) + \partial_{\bDelta} S_N(\tilde{\rho}, \tilde{\bDelta})^\top (\hat{\bDelta} - \bDelta^*) \big).
    \end{align}
    In the remainder of the proof, we show that the denominator is bounded away from zero with probability $1-O(e^{-cN})$, and that the numerator exhibits exponential concentration.

    \paragraph{Step 2. Bounding the denominator of \eqref{eq:score equation}.}
    We claim that $\partial_\rho S_N(\tilde{\rho}, \tilde{\bDelta}) \le -\kappa(\rho^*)/3$ with probability greater than $1 - O(e^{-cN})$. By Step 0, we can work on the event $E_1, E_2$ (see \eqref{eq:event_delta}). For notational convenience, denote $J_N := \partial_\rho S_N$. As $J_N$ is a continuous function on a compact set, we can write
    $$|J_N(\tilde{\rho},\tilde{\bDelta}) - J_N(\rho^*,\bDelta^*)| \le M_3|\tilde{\rho}-\rho^*| + M_{\bDelta} \|\tilde{\bDelta} - \bDelta^*\| \le \frac{\kappa(\rho^*)}{3}.$$
    Here, $M_3, M_{\bDelta}$ denote the upper bounds of $\partial_{\rho}^3 \ell, \partial_\rho^2 \partial_{\bDelta} \ell$ as introduced in the beginning of the proof, and the final inequality follows as we are working under $E_1, E_2$.

    Also, noting that $\EE J_N(\rho^*,\bDelta^*) = \ell_\infty''(\rho^*) \le -\kappa(\rho^*)$ (see \eqref{eq:kappa def}), Hoeffding's inequality gives that
    $$\PP\Big(|J_N(\rho^*,\bDelta^*) - \ell_\infty''(\rho^*)|> t \Big) \lesssim e^{-c N t^2}.$$
    By taking $t = \kappa(\rho^*)/3$ and combining everything, we get
    $$J_N(\tilde{\rho},\tilde{\bDelta}) \le J_N(\rho^*,\bDelta^*) + \frac{\kappa}{3} \le \ell_\infty''(\rho^*) + \frac{2\kappa}{3} \le -\frac{\kappa}{3}$$
    with probability greater than $1-O(e^{-cN})$. The constant $c$ here can be made uniform in $\rho^*$ by lower bounding $\kappa = \kappa(\rho^*)$ with a constant that only depends on $\delta_1, M$.

    \paragraph{Step 3. Concentration for the numerator of \eqref{eq:score equation}.}
    We separately control the two terms in the numerator. For the first term, we use the following identity:
    $$\EE \partial_\rho \log \pi_{R_{i,1}, R_{i,2}}(\rho^*; \bDelta^*) = 0.$$
    This follows from using the definition of KL divergence to note that $\ell_\infty$ is maximized at $\rho^\star$, and re-writing $\partial_\rho \ell_\infty (\rho^\star) = 0$.
    Writing $S_N$ as the i.i.d. sum of bounded random variables:
    $$S_N(\rho^*,\bDelta^*) = \frac{1}{N} \sumin \partial_\rho \log \pi_{R_{i,1}, R_{i,2}}(\rho^*; \bDelta^*),$$
    Hoeffding's inequality gives
    $$\PP(|S_N(\rho^*,\bDelta^*)| > t) \lesssim e^{-cN t^2}.$$

    For the second term in the numerator of \eqref{eq:score equation}, we use the concentration of $\hat{\bDelta} - \bDelta$ in \cref{prop:pi delta consistency} and the fact that the derivatives $\partial_{\bDelta} S_N$ are bounded in the compact parameter space to get
    $$\PP\Big(|\partial_{\bDelta} S_N(\tilde{\rho}, \tilde{\bDelta})^\top (\hat{\bDelta} - \bDelta^*)| > t \Big) \le \PP \Big( \|\hat{\bDelta} - \bDelta^*\| > ct \Big) \lesssim e^{-cNt^2}.$$
    Note that we omit the dependence of $C_1+C_2$ in both inequalities.

    The final conclusion is immediate by combining steps 2 and 3.
\end{proof}

\begin{proof}[Proof of \cref{lem:lower bound of pi}]
    (a) Writing out the p.d.f. of $(X_1,X_2)$ as
    $$\phi_2(x_1,x_2;\rho) = \frac{1}{2\pi\sqrt{(1-\rho^2)}} e^{-\frac{x_1^2 - 2 \rho x_1 x_2 + x_2^2}{2(1-\rho^2)}} \ge \frac{1}{2\pi} e^{-\frac{x_1^2 - 2 \rho x_1 x_2 + x_2^2}{2(1-\rho^2)}},$$
    it suffices to lower bound the exponential term.

    For each category interval, choose a subinterval of length at least $\delta_*:=\min\{\delta_2,1\}>0$ that lies inside the category interval and inside $[-M-1,M+1]$. For interior categories, this follows from the separation condition; for the first and last categories, use the intervals adjacent to the finite boundary threshold. For any $(x_1,x_2)$ in the resulting rectangle, we have
    \begin{align*}
        x_1^2 - 2\rho x_1 x_2 + x_2^2&\le 4(M+1)^2, \quad 1-\rho^2 \ge 1 - (1-\delta_1)^2 \ge \delta_1,
    \end{align*}
    so $-\frac{x_1^2 - 2 \rho x_1 x_2 + x_2^2}{2(1-\rho^2)} \ge -\frac{2(M+1)^2}{\delta_1}$.
    Then, the lower bound of $\pi_{a,b}$ follows by integrating over this subrectangle:
    \begin{align*}
        \pi_{a,b}(\rho;\bDelta_1,\bDelta_2) = \int_{\Delta_{1,a-1}}^{\Delta_{1,a}} \int_{\Delta_{2,b-1}}^{\Delta_{2,b}} \phi_2(x_1,x_2;\rho) dx_2 dx_1 \ge \frac{\delta_*^2}{2\pi} e^{-\frac{2(M+1)^2}{\delta_1}} := c > 0.
    \end{align*}

    \noindent (b) The lower bound $0$ is immediate as $\pi_{a,b}$ (as a function of $\Delta_{1,a}$) is increasing in $\Delta_{1,a}$. By the chain rule and using the conclusion of part (a), we have
    \begin{align*}
        \frac{\partial \log \pi_{a,b}(\rho; \bDelta)}{\partial \Delta_{1,a}} &= \frac{1}{\pi_{a,b}(\rho; \bDelta)} \frac{\partial \pi_{a,b}(\rho; \bDelta)}{\partial \Delta_{1,a}} \le \frac{1}{c} \frac{\partial \pi_{a,b}(\rho; \bDelta)}{\partial \Delta_{1,a}} \\
        &= \frac{1}{c} \int_{\Delta_{2,b-1}}^{\Delta_{2,b}} \phi_2(\Delta_{1,a},x_2;\rho)dx_2 \le \frac{\phi(\Delta_{1,a})}{c} \le \frac{1}{\sqrt{2\pi} c}.
    \end{align*}
    Here, the equality in the second line used the fundamental theorem of calculus, and the inequalities follow from relaxing the integration range. This completes the proof as we can take $L = 1/(\sqrt{2\pi} c)$. For other derivatives (such as with respect to $\Delta_{1,a-1}$ when $a>1$), the partial derivative is still absolutely bounded by $L$.
\end{proof}

\section{Simulation Details}\label{supp:simulation-details}
\subsection{Data generation details}\label{supp:sim data generation}
To generate the data for each simulation run, we first construct the block-diagonal precision matrices $\bOmega_{k}$. For each block, let $\mathbf B$ be the adjacency matrix of the Erd\H{o}s-R\'enyi random graph with probability $s$. Set each block precision matrix as $\bOmega_{k}^{(\ell)} = 0.5 \mathbf B + \delta \II$, where $\delta$ is dynamically chosen such that the condition number of the matrix is exactly equal to the block dimension $M = J/L$. %
Finally, we invert $\bOmega_{k}^{(\ell)}$ to get the raw covariance matrix and rescale it so that $\bSigma_k$ is a correlation matrix with unit diagonals.

For each observation $i$, assume the class proportions are equal ($\pi_k = 1/K$), and generate a latent class $Z_i$ uniformly from $[K]$. Given a latent class assignment $Z_i = k$, we draw a continuous latent vector $\XX_i \sim \mathcal{N}(\mathbf 0, \bSigma_{k})$. We then discretize $\XX_i$ into the observed ordinal response vector $\RR_i$ using class-specific threshold parameters. For the binary items, we draw a single threshold $\Delta_{j,1,k} \sim U[-1, 1]$. For the four-category items, we draw three sorted thresholds $(\Delta_{j,1,k}, \Delta_{j,2,k}, \Delta_{j,3,k})$ uniformly from $[-2, 2]$.

\subsection{Implementation choices for Step 3}\label{supp:sim tuning parameter cross validation}

\paragraph{Precision matrix estimation}
We implement the graphical lasso to estimate the precision matrices. When the covariance matrix estimate $\hat{\bSigma}_k^{(\ell)}$ is not positive definite, we ensure this by appropriately increasing the diagonal elements.
Later, in \cref{supp:comparison clime}, we conduct additional simulations comparing the graphical lasso with CLIME. They have similar accuracy, whereas the graphical lasso is much faster.

\paragraph{Tuning parameter selection}
For each simulation setting, we run 10 preliminary simulations to select the tuning parameter $\lambda$ via cross-validation, as mentioned in \cref{subsec:additional choices}. We choose the most frequently selected value, which is often identical for all 10 trials. This selected $\lambda$ is fixed throughout the main trials (100 replications). Below, we elaborate on the tuning parameter selection.

We randomly partition the dataset into $F=5$ folds of roughly equal sizes. For a given candidate tuning parameter $\lambda$ from the grid $\Lambda = \{0.1, 0.5, 1, 2\} \times \sqrt{\frac{\log(JK)}{N}}$, we iteratively hold out each fold $f \in \{1, \dots, F\}$ as the test set and use the remaining data as the training set. Let $\hat{\bSigma}_{k}^{(f)}$ be the estimated covariance matrix (estimated via Step 2) using only the observations from the $f$-th fold, and let $\hat{\bOmega}_{k, \lambda}^{(-f)}$ denote the precision matrix estimate using the observations excluding the $f$-th fold (via Step 3).

We evaluate the performance of each candidate $\lambda$ using the predictive negative Gaussian log-likelihood. The optimal tuning parameter is chosen by minimizing this validation loss aggregated across all $K$ classes and $F$ test folds:
$$ \text{CV}(\lambda) = \sum_{f=1}^{F} \sum_{k=1}^K \hat{N}_k^{(f)} \left[ \text{tr}\left(\hat{\bSigma}_k^{(f)} \hat{\bOmega}_{k,\lambda}^{(-f)}\right) - \log \det\left(\hat{\bOmega}_{k,\lambda}^{(-f)}\right) \right], $$
where $\hat{N}_k^{(f)}$ denotes the number of observations assigned to the $k$-th cluster when holding out the $f$-th fold. After selecting the optimal tuning parameter $\hat{\lambda} = \arg\min_{\lambda \in \Lambda} \text{CV}(\lambda)$, we finally run Step 3 to estimate $\hat{\bOmega}$.

\subsection{Guideline for selecting $K$}\label{supp:select-K}
In the simulations, we assume that $K$ is correctly specified. In general, the number of latent classes, $K$, is often unknown. In spectral clustering, one standard approach is the ``elbow method,'' which uses a scree plot \citep{thorndike1953belongs,ketchen1996application}. This procedure involves evaluating the k-means objective function \eqref{eq:k-means} across a range of $K$ values and selecting the ``elbow'' value where the decrease in clustering error flattens. Similarly, one can observe the largest gap/ratio among the singular values of the flattened data matrix $\tilde{\RR}$. 
Alternative methods include using information criteria (such as BIC or Extended BIC) after fitting the entire model for each $K$, or evaluating pointwise cohesion and separation via the Silhouette score \citep{rousseeuw1987silhouettes}.

\section{Additional Comparisons and Visualizations}\label{supp:comparison-visualization}

\subsection{Improving clustering accuracy via likelihood refinement}\label{sec:em comparison} %
We discuss and evaluate a refined clustering method based on the learned block structures. Note that our original Step 1 is a simple spectral clustering procedure that does not take into account the dependence structure within the data. Hence, we propose a one-step likelihood refinement that uses the estimated threshold vector $\hat{\bDelta}$, the block partition $\hat{V}$, and the blockwise covariance matrices $\hat{\bSigma}_k^{(\ell)}$. The refinement updates the class assignment by maximizing the blockwise ordinal likelihood:
$$
\hat{Z}_i^{\,\mathrm{ref}}=\argmax_{1\le k\le K} \prod_{\ell=1}^{\hat{L}}\PP\left(\XX^{(\ell)} \in \prod_{j \in \hat{V}_\ell} \Bigl[\hat{\Delta}_{j,R_{i,j}-1,k}, \hat{\Delta}_{j,R_{i,j},k}\Bigr)
\,\middle|\,\XX^{(\ell)} \sim N\left(\mathbf 0,\hat{\bSigma}_k^{(\ell)}\right) \right),$$
where the product term denotes the rectangle corresponding to the observed ordinal response pattern $R_{i,j}$ in block $\ell$.

We fix $N=1000$ and $J=100$, and vary $K,L$, and $s$. Here $s$ denotes the connectivity probability used when generating the precision matrix. To create a more challenging clustering setting, we reduce the separation among the threshold parameters compared to the values in \cref{supp:sim data generation}. Specifically, after generating the original threshold vectors $\bDelta_k$ as in \cref{supp:sim data generation}, we rescale them toward their across-class mean $\bar{\bDelta}=K^{-1}\sum_{k=1}^K \bDelta_k$ by setting
\[
\bDelta_k' = \frac{1}{3} \bar{\bDelta} + \frac{2}{3} \bDelta_k.
\]
We compare this one-step likelihood refinement against our Step 1 (the spectral clustering method without the refinement) and the usual locally independent LCM fitted by EM. For fair comparison, both the EM algorithm and the k-means in Step 1 use 10 random starts and a maximum of 100 iterations.

\begin{table}[h!]
\centering
\begin{tabular}{cccc|ccc|ccc}
\toprule
\multirow{2}{*}{$J$}
& \multirow{2}{*}{$K$}
& \multirow{2}{*}{$L$}
& \multirow{2}{*}{$s$}
& \multicolumn{3}{c|}{Err($\hat{\ZZ}$)}
& \multicolumn{3}{c}{RMSE($\hat{\bpi}$)} \\
& & &
& Spec. & Refined & EM
& Spec. & Refined & EM \\
\midrule
\multirow{8}{*}{$100$}
& \multirow{4}{*}{$5$}
& \multirow{2}{*}{$5$}
& 0.6 & 4.2 & 0.9 & 1.5 & 0.013 & 0.013 & 0.013 \\
& & & 0.2 & 4.6 & 0.5 & 1.3 & 0.013 & 0.013 & 0.013\\
& & \multirow{2}{*}{$10$}
&     0.6 & 4.0 & 0.4 & 1.4 & 0.011 & 0.011 & 0.011 \\
& & & 0.2   & 4.2 & 0.3 & 1.6 & 0.013 & 0.013 & 0.013 \\
& \multirow{4}{*}{$10$}
& \multirow{2}{*}{$5$}
& 0.6 & 19.2 & 9.3 & 12.9 & 0.009 & 0.009 & 0.013 \\
& & & 0.2   & 23.0 & 6.9 & 11.9 & 0.009 & 0.009 & 0.013 \\
& & \multirow{2}{*}{$10$}
& 0.6 & 18.7 & 4.6 & 9.0 & 0.009 & 0.009 & 0.011 \\
& & & 0.2   & 18.3 & 2.9 & 11.6 & 0.009 & 0.009 & 0.013 \\
\bottomrule
\end{tabular}
\caption{Comparison of spectral clustering (Step 1), the one-step likelihood refinement, and the usual EM algorithm under local independence. %
}
\label{tab:lcm clustering comparison}
\end{table}

The results are reported in \cref{tab:lcm clustering comparison}. The one-step likelihood refinement substantially improves the initial spectral clustering labels across all considered settings. For instance, when $K=10$, where the spectral initialization has noticeably larger clustering error, the refinement reduces the error by more than half. This supports the use of the estimated blockwise covariance structure after Step 1. This advantage is less apparent in terms of estimating the proportion parameters, as the RMSE values do not decrease.

Another interesting observation is that the naive EM algorithm performs reasonably well for clustering, despite being misspecified. We believe this is because the coordinate-wise likelihood $\PP(R_{i,j} \mid Z_i, \bDelta_j)$ remains somewhat robust under moderate local dependence, and the misspecified marginal likelihood may be informative for estimating the threshold vector $\bDelta$. In fact, one may also view the spectral clustering objective as a misspecified independent Gaussian likelihood. A related phenomenon of a misspecified MLE having satisfactory performance under Gaussian mixtures has been investigated in \cite{lee2025inference}. We note that the one-step refinement is better than EM in the reported settings, suggesting that incorporating the estimated dependence structure is indeed more beneficial.

\subsection{Comparison of graphical model selection methods}\label{supp:comparison clime}
We conduct simulations by considering two procedures for the precision matrix in Step 3: graphical lasso and CLIME, while fixing $J = 100$. 
The results in \cref{tab:clime vs glasso} show that CLIME and graphical lasso lead to broadly comparable accuracy for estimating the sparse precision matrices. Across the eight settings, the support recovery error $\mathrm{Err}_0(\hat{\bOmega})$ and Frobenius error $\mathrm{Err}_F(\hat{\bOmega})$ are of similar magnitude. In several settings CLIME gives slightly smaller errors, especially when $L=10$. The ROC curves in \cref{fig:roc clime} also resemble \cref{fig:roc}, and we draw similar conclusions when Steps 1 or 2 are omitted.

The main difference is computational: CLIME is noticeably slower than graphical lasso. The computational gap is more evident when $L$ and $K$ increase. For example, when $(K,L)=(10,10)$, CLIME takes about $30$ seconds on average, whereas graphical lasso takes only about $4$ seconds. The computational burden of CLIME further increases in larger dimensions, say $J = 500$.

\begin{table}[h!]
\centering
\begin{tabular}{cccc|ccc|ccc}
\toprule
    \multirow{2}{*}{$J$} 
    & \multirow{2}{*}{$K$} 
    & \multirow{2}{*}{$L$} 
    & \multirow{2}{*}{$s$} 
    & \multicolumn{3}{c|}{glasso} 
    & \multicolumn{3}{c}{CLIME} \\
    & & & 
    & $\text{Err}_0(\hat{\bOmega})$ & $\text{Err}_F(\hat{\bOmega})$ & time (s)
    & $\text{Err}_0(\hat{\bOmega})$ & $\text{Err}_F(\hat{\bOmega})$ & time (s) \\
    \midrule
    \multirow{8}{*}{$100$}   & \multirow{4}{*}{$5$}  & \multirow{2}{*}{$5$}  & 0.6 & 0.084 & 0.104 & 2.4 & 0.088 & 0.111 & 9.8 \\
                            &                       &                       & 0.2 & 0.085 & 0.100 & 2.3 & 0.086 & 0.104 & 9.5 \\
                            &                       & \multirow{2}{*}{$10$} & 0.6 & 0.035 & 0.088 & 2.5 & 0.034 & 0.084 & 16.5 \\
                            &                       &                       & 0.2 & 0.038 & 0.086 & 2.3 & 0.037 & 0.081 & 16.0 \\
                            & \multirow{4}{*}{$10$} & \multirow{2}{*}{$5$}  & 0.6 & 0.090 & 0.112 & 4.1 & 0.091 & 0.108 & 18.6 \\
                            &                       &                       & 0.2 & 0.068 & 0.111 & 3.9 & 0.062 & 0.103 & 18.5 \\
                            &                       & \multirow{2}{*}{$10$} & 0.6 & 0.045 & 0.088 & 4.1 & 0.041 & 0.093 & 30.7 \\
                            &                       &                       & 0.2 & 0.038 & 0.096 & 4.1 & 0.034 & 0.090 & 30.2 \\
\bottomrule
\end{tabular}
\caption{Comparison of graphical lasso and CLIME for precision matrix estimation and computation time. Here, Steps 1 and 2 are identically implemented.}
\label{tab:clime vs glasso}
\end{table}

\begin{figure}[h!]
    \centering
    \includegraphics[width=\linewidth]{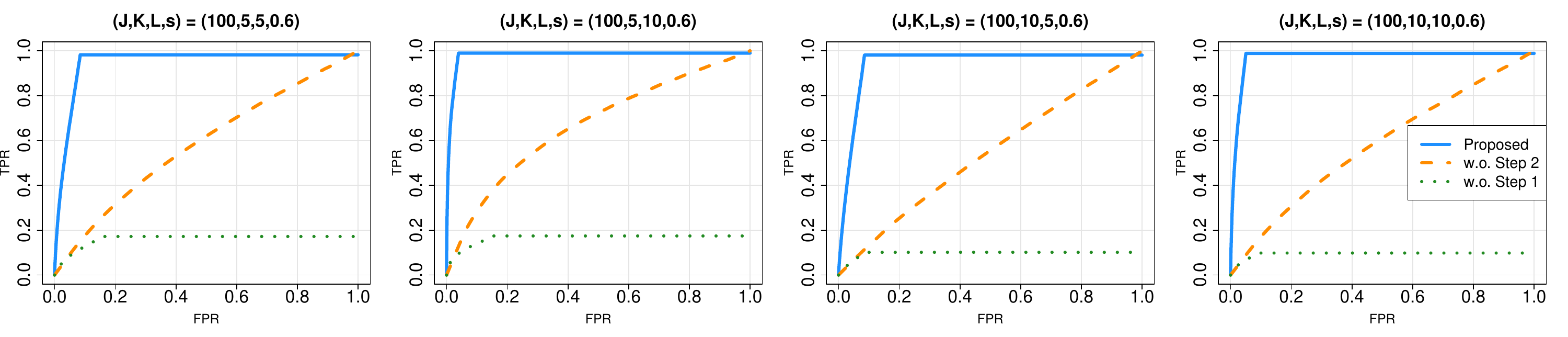}
    \includegraphics[width=\linewidth]{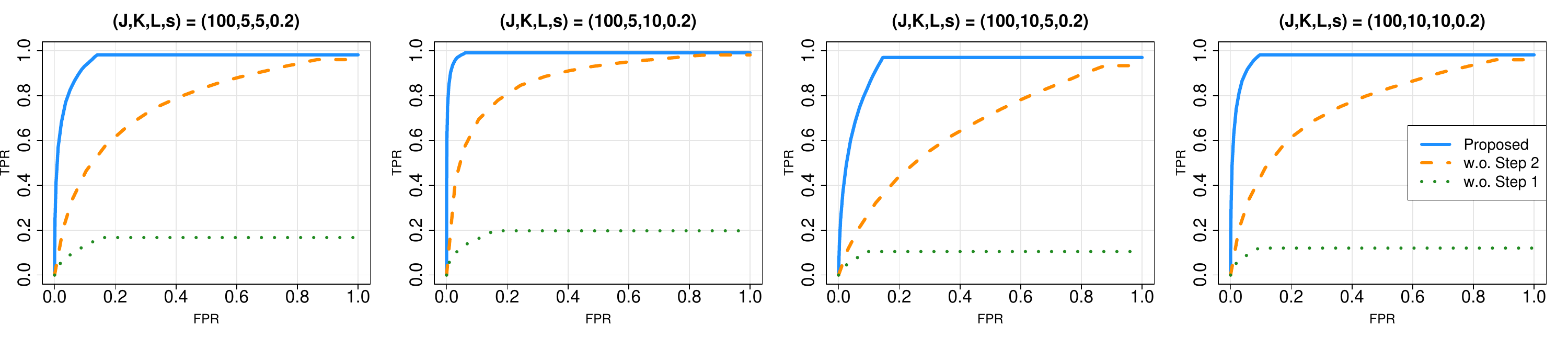}
    \caption{ROC curves for selecting $\bOmega$ under $J = 100$ when the CLIME estimator is used. The solid line corresponds to the proposed method, whereas the dashed/dotted line corresponds to variants where Steps 2/1 are omitted, respectively.}
    \label{fig:roc clime}
\end{figure}

\subsection{Illustration of incorrectly estimated blocks}\label{sec:robust}
We visualize a typical instance of incorrect block recovery in \cref{fig:block recovery wrong}. The figure is analogous to \cref{fig:block recovery}, but instead considers a more challenging simulation setting with lower signal and smaller block length. We see that $\hat{L} = 9$ is underestimated because two of the true blocks are merged. In general, we observed $\hat{L} \le L$ as discussed in the main text (see \cref{sec:sims}). The underestimation of blocks potentially leads to a higher FDR, but is better in terms of TPR as opposed to overestimating (i.e., further partitioning a correct block).

\begin{figure}[h!]
    \centering
    \includegraphics[width=0.6\linewidth]{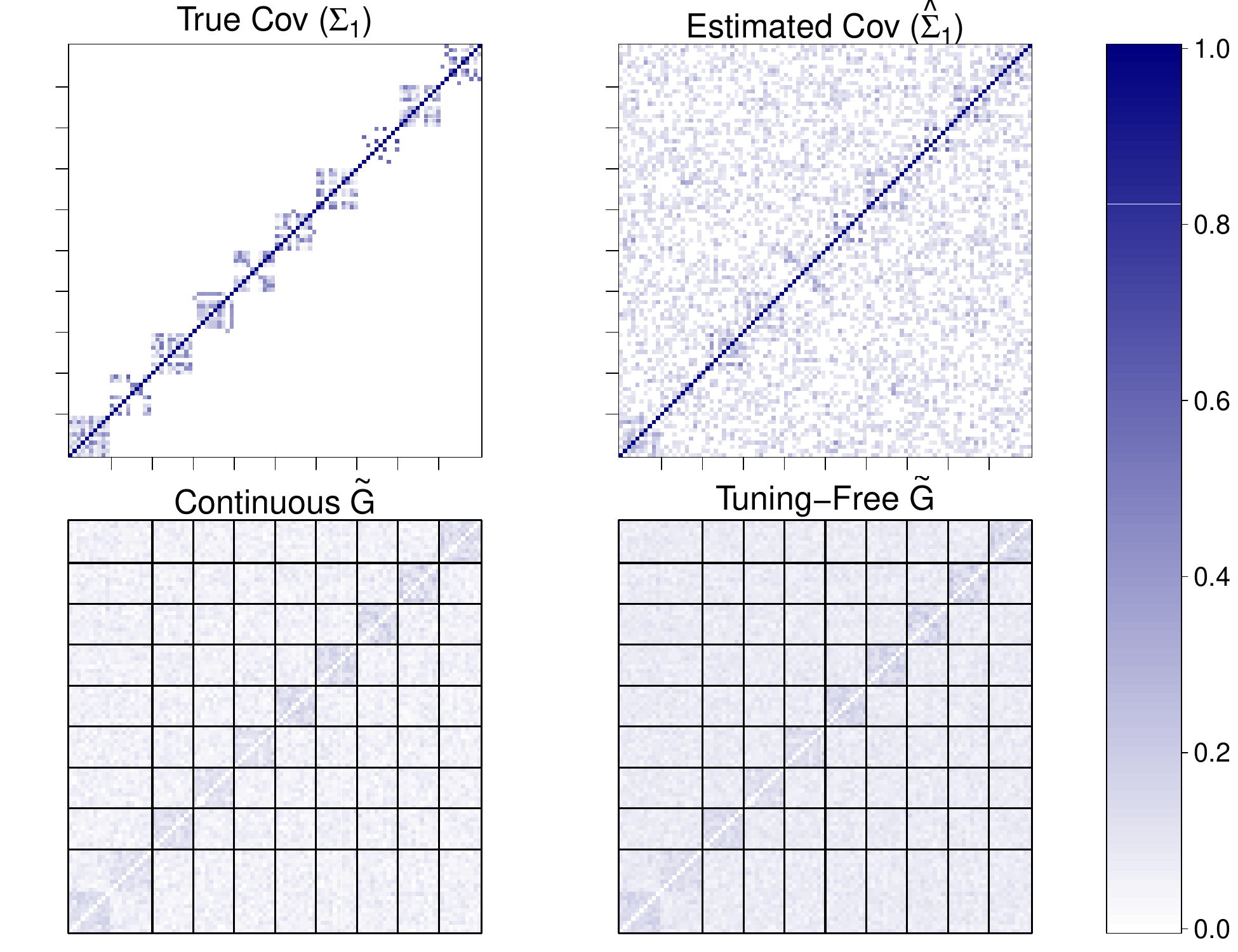}
    \caption{Analogue of \cref{fig:block recovery} under the challenging setting with parameters as in the 8th row of \cref{tab:sim result}. While the true number of blocks is $L = 10$, the boundary between the first and second block is unclear. Thus, the Leiden algorithm selects $\hat{L} = 9$ blocks where the true blocks indexed by $\ell = 1,2$ are merged.}
    \label{fig:block recovery wrong}
\end{figure}

\subsection{Tuning parameter sensitivity of our method and JGL}\label{supp:sensitivity}
\begin{figure}[h!]
    \centering
    \includegraphics[width=0.32\linewidth]{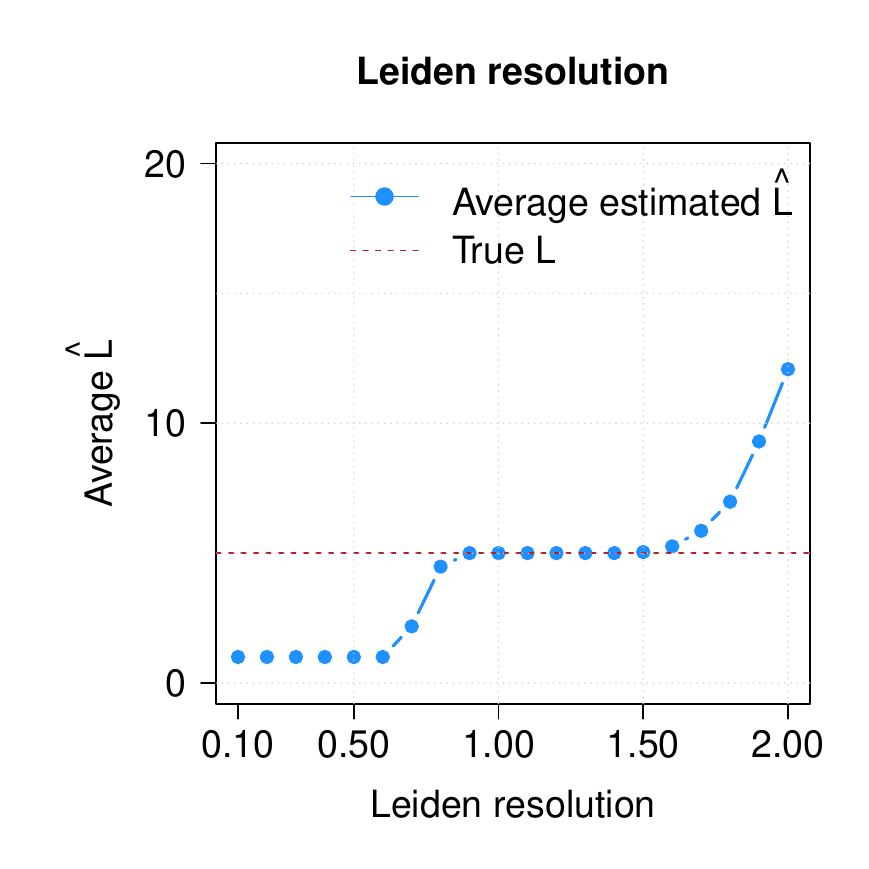}
    \includegraphics[width=0.32\linewidth]{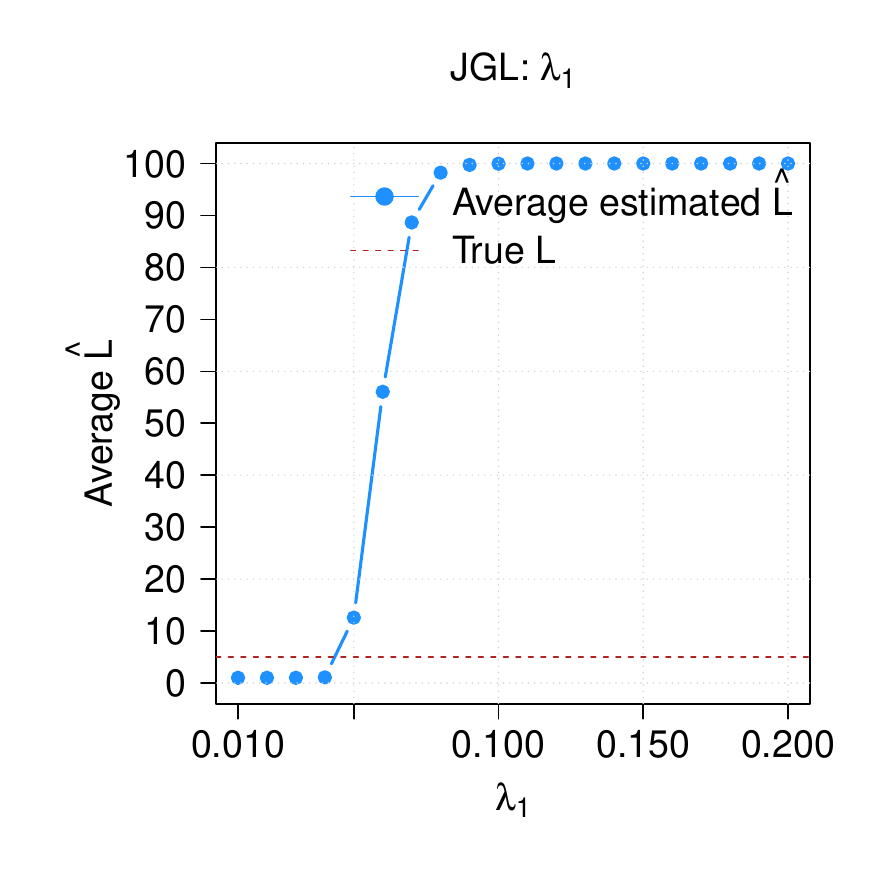}
    \includegraphics[width=0.32\linewidth]{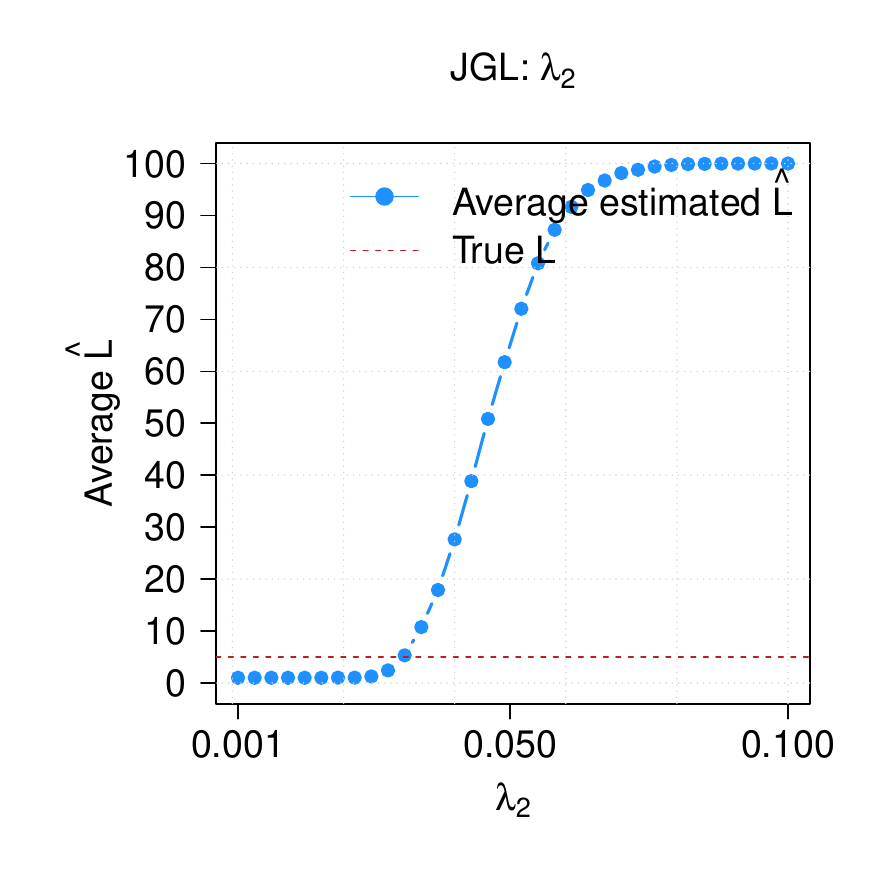}
    \caption{Sensitivity of the estimated number of blocks to tuning parameters. The left panel varies the Leiden resolution parameter, and the center/right panels vary the JGL parameters $\lambda_1$/$\lambda_2$ (with the other parameter fixed at $\lambda_2 = 0.001$ and $\lambda_1 = 0.02$, respectively). The horizontal dashed line marks the true number of blocks, $L=5$. JGL is highly sensitive to the tuning choice, whereas the proposed Leiden-based block recovery is more robust to its resolution parameter.}
    \label{fig:tuning}
\end{figure}

\noindent In \cref{fig:tuning}, we plot the average $\hat{L}$ with respect to the tuning parameter, for both our method and JGL. We set $K=L=5, ~s=0.6$. Our method is robust to the choice of resolution parameter, for a wide range of values from 0.8 to 1.5, which includes the default choice of 1. In other words, the proposed method performs well at a wide range of Leiden resolution parameters. In contrast, JGL is highly sensitive to its tuning parameters, and leads to uninformative block structures with $\hat{L} = 1$ or $J=100$ for most values. We believe this instability is because the block structure learned by JGL is highly sensitive to false discovery. To elaborate, off-block-diagonal entries $\sigma_{(j,j'),k}$ that are larger than the threshold in Theorem 2 in \cite{danaher2014joint} (which depends on tuning parameters) result in merging the corresponding blocks.

This sensitivity made fully data-driven JGL tuning difficult in our experiments. Cross-validation and information criteria selected extreme block structures: CV, AIC, and BIC selected $\hat{L}=1$, whereas EBIC selected $\hat{L}=J$. For example, CV selected $\lambda_1 = 0.02, \lambda_2 = 0.0017$ which resulted in $\hat{L} = 1$ blocks, which is why we additionally reported results under a more fine-tuned penalty in \cref{tab:sim comparison jgl}. Fixing $\lambda_2 = 0.001$ (larger values resulted in worse ARI), we chose $\lambda_1 = 0.047$ so that the average $\hat{L}$ is $5.1$.

\section{Real Data Analysis Details}\label{supp:real-data-details}
\subsection{Implementation details}
We provide further details on the implementation of each step of our method. 

\begin{figure}[h!]
    \centering
    \includegraphics[width=0.7\linewidth]{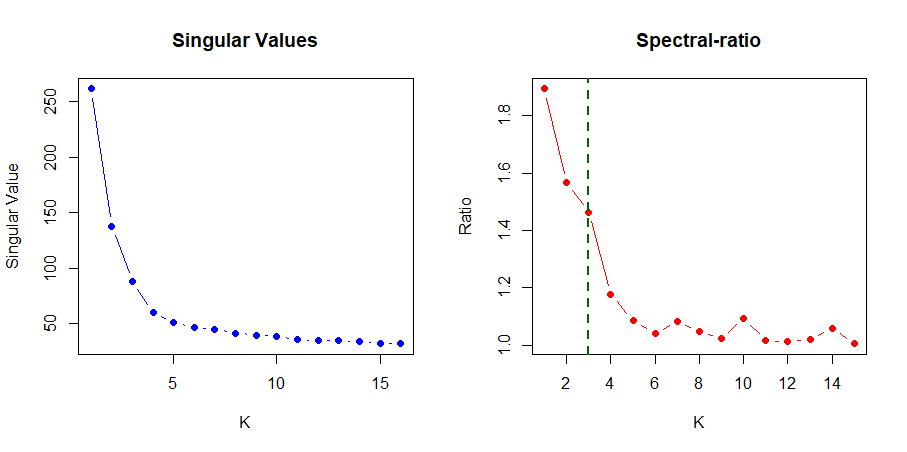}
    \caption{Visualization of the singular values and spectral ratios in the ANES data.}
    \label{fig:anes-spectral-gap}
\end{figure}

\paragraph{Preprocessing}
 We preprocess the ANES data following \cite{chen2024generalized}. We additionally remove one item that directly asks about the respondent's party. The preprocessing details for the HapMap data are given in the main text.

\paragraph{Step 1}
For the ANES data, we selected $K = 3$ mainly to match the interpretation of the three latent classes with political ideology. This is also supported by the spectral ratio, as $K=3$ corresponds to a rough elbow point in \cref{fig:anes-spectral-gap}. For the HapMap data, we selected $K = 4$ to match the four subpopulations considered (Utah residents with European ancestry, Mexican ancestry in Los Angeles, Gujarati Indians in Houston, and Yoruba in Ibadan, Nigeria). Due to the numeric encoding of the responses (0, 1, 2), we implemented the clustering in \cref{algo:spectral} directly on the raw data matrix $\RR$ without flattening. Note that we slightly abuse the notation, since the typical sample space for $R_j$ with 3 categories is $\{1,2,3\}$.

\paragraph{Step 2}
In Step 2, we used the Leiden algorithm on the tuning-free version of the aggregated weight matrix $\tilde{\GG}$. The high-resolution parameter value 1.5 is chosen based on the sensitivity analysis in \cref{fig:tuning}.

\paragraph{Step 3 (for the ANES data)}
In Step 3, we implemented graphical lasso. The tuning parameter $\lambda = 0.3 \sqrt{\log(JK)/N} = 0.019$ was selected by cross-validation, from the grid $\Lambda = \{0.1, 0.3, 0.5, 1, 2\} \times \sqrt{\log(JK)/N}$. We also tried selecting a separate $\lambda_k$ per latent class, and the above value was selected for all classes.

\subsection{Additional visualization for the ANES data}
We visualize the histogram of the number of categories $C_j$ for all items in the ANES data in \cref{fig:anes_histogram}. We see that $C_j = 5$ is the most common, which typically takes ordered response categories ``Not at all''-``A little''-``Somewhat''-``Very''-``Extremely'', or ``Extremely important''-``Very important''-``Moderately important''-``Slightly important''-``Not at all important''.
Items with $C_j = 7$ typically arise when the voters are asked to rate on a scale of 1-7 without specific categories for the values 2-6. We also provide the details of the baseline block structures in \cref{tab:anes true block}, which were use to compute the ARI values for the estimated blocks. The block names provided here were also used as a rough guideline for interpreting the estimated blocks in \cref{tab:cluster_hierarchy}.

\begin{figure}[h!]
    \centering
    \includegraphics[width=0.4\linewidth]{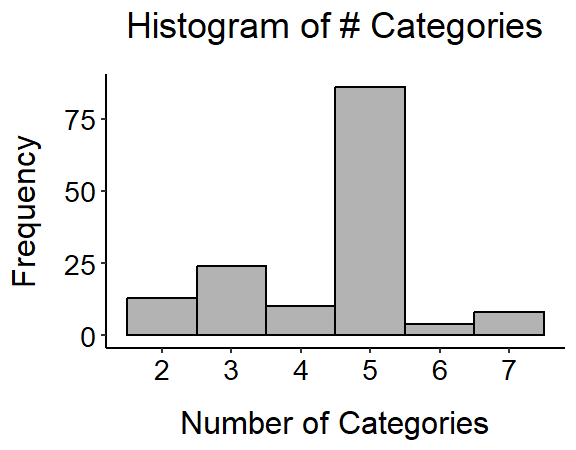}
    \caption{Histogram of the number of response categories $C_j$ in the ANES dataset.}
    \label{fig:anes_histogram}
\end{figure}

\begin{table}[h!]
    \centering
    \begin{tabular}{ccc|ccc}
    \toprule
        $\ell$ & $J_\ell = |V_\ell|$ & Tag & $\ell$ & $J_\ell = |V_\ell|$ & Tag \\
    \midrule
        1 & 1 & Follow politics & 17 & 3 & Guns and crime \\
        2 & 1 & Voter registration & 18 & 6 & Immigrant emotions \\
        3 & 1 & Turnout & 19 & 9 & Democratic attitudes \\
        4 & 7 & Participation & 20 & 12 & Electoral integrity \\
        5 & 6 & Global emotion & 21 & 1 & Political efficacy \\
        6 & 3 & Presidential approval & 22 & 1 & Racism \\
        7 & 3 & Economic performance & 23 & 1 & Feminist attitudes \\
        8 & 8 & Inflation & 24 & 3 & Political tolerance \\
        9 & 14 & Issue importance & 25 & 11 & Racial stereotypes \\
        10 & 12 & Issue ownership & 26 & 2 & Identities \\
        11 & 1 & Climate & 27 & 5 & Role of schools \\
        12 & 1 & Ukraine & 28 & 1 & Great replacement \\
        13 & 4 & Trust experts & 29 & 4 & Racial privilege \\
        14 & 2 & Political disagreement & 30 & 2 & Transgender attitudes \\
        15 & 9 & Abortion & 31 & 4 & Racial resentment \\
        16 & 6 & Abortion emotions & - & - & - \\
    \bottomrule
    \end{tabular}
    \caption{{The baseline block structure (with $L = 31$ blocks) in the ANES questionnaire.}}
    \label{tab:anes true block}
\end{table}

Next, in \cref{tab:anes item matching}, we provide the exact item description for the 11 nodes that appear in \cref{fig:precision_anes}. While plotting \cref{fig:precision_anes}, we omitted edges whose weights are too small (defined by $|\omega_{(j,j'),k}| < 3 \lambda$) for better visibility. Without this thresholding, each $\hat{\bOmega}_k$ has 41, 40, and 41 edges, of which 20 are shared. We also visualize the full $\hat{\bOmega}_k$ (the $144 \times 144$ matrix) in \cref{fig:anes_omega}, which illustrates the overall sparsity of the estimated precision matrices. Note that a negative precision matrix entry $\omega_{(j,j'),k} < 0$ (colored in red) implies a positive partial correlation $\rho_{(j,j'),k}$, which is why the majority of the off-diagonal entries are red.
\begin{figure}
    \centering
    \includegraphics[width=0.9\linewidth]{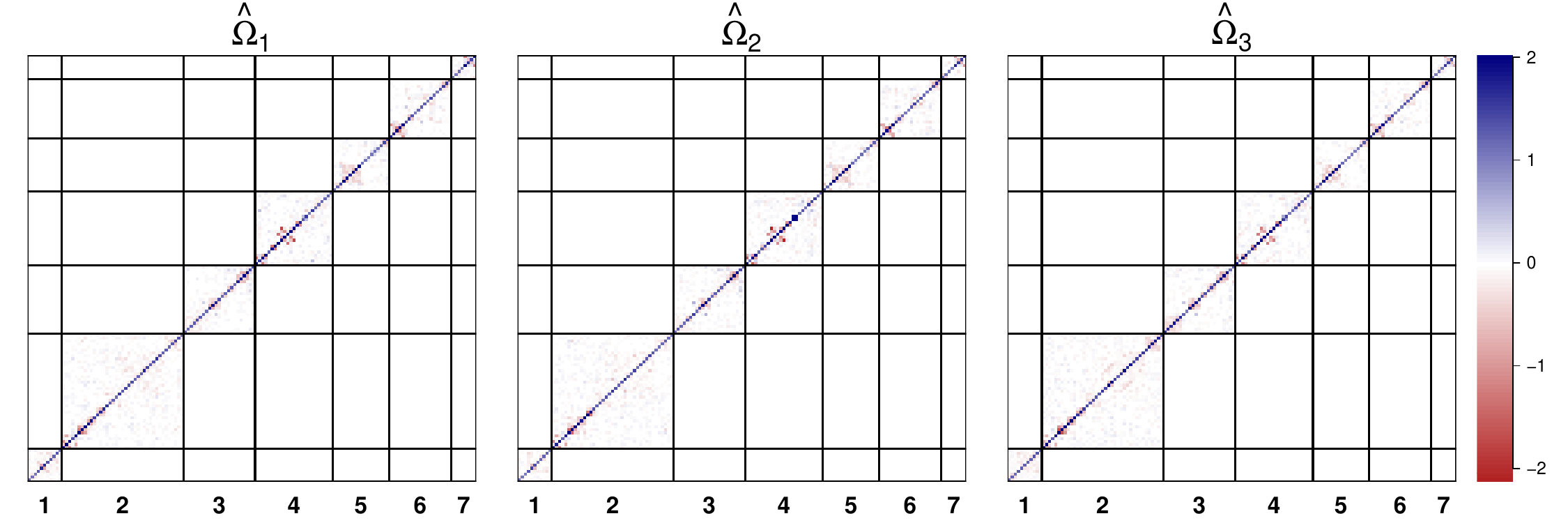}
    \caption{Visualization of the estimated precision matrices $\hat{\bOmega}_k$ for the ANES data. The elements are thresholded at $\pm 2$ for better illustration.}
    \label{fig:anes_omega}
\end{figure}

\begin{table}[h!]
\centering
\begin{tabular}{c l p{11cm}}
\toprule
\textbf{Index} & \textbf{Variable Name} & \textbf{Questionnaire item} \\
\midrule
1  & victimcrimebl  & How worried are you about black citizens being the victims of violent crime? \\
2  & victimpolicebl & How worried are you about black citizens being hurt or killed by the police? \\
3  & rprob          & In the United States, how serious a problem is racism? \\
4  & nonwhite3      & The demographic makeup of America is changing, with the U.S. Census Bureau estimating white people will become a minority. How do you feel? \\
5  & whpriv         & In American society, do you think that being White comes with advantages? \\
6  & blpriv         & In American society, do you think that being Black comes with advantages? \\
7  & hipriv         & In American society, do you think that being Hispanic comes with advantages? \\
8  & aspriv         & In American society, do you think that being Asian comes with advantages? \\
9  & rr1            & Irish, Italians, Jewish and many other minorities overcame prejudice and worked their way up. Blacks should do the same without any special favors. 
 \\
10 & rr2            & Generations of slavery and discrimination have created conditions that make it difficult for Blacks to work their way out of the lower class. \\
11 & rr3            & Over the past few years, Blacks have gotten less than they deserve. \\
\bottomrule
\end{tabular}
\caption{Item identifiers and details for the visualization in \cref{fig:precision_anes}.}
\label{tab:anes item matching}
\end{table}

\subsection{Additional discussion for the HapMap3 data}\label{supp:hapmap discussion}
We provide additional context for the problem of learning block structures in genetics. It is well understood that closely linked SNPs tend to exhibit stronger correlations \citep{gabriel2002structure}. Commonly known as linkage disequilibrium (LD), this phenomenon naturally partitions the human genome into ``haplotype blocks'' of dependent SNPs, and learning such blocks has been an active area of research \citep{berisa2015approximately,yoo2015clique,kim2018new}. 

A common empirical strategy is to cluster SNPs using pairwise Pearson correlations, $\text{Cor}(R_j, R_{j'})$. This can be problematic in the presence of latent classes, as the correlations induced by subpopulation differences may be mistaken for LD. Our analysis illustrates that the proposed method can learn block structure while accounting for unknown subpopulations. We note that some analyses estimate separate block structures for each subpopulation \citep{berisa2015approximately}. However, haplotype block boundaries are often similar across populations \citep{gabriel2002structure}, which motivates our use of a shared block structure. Finally, we note that our modeling does not enforce the Hardy-Weinberg equilibrium (i.e., it does not assume that $R_j \in \{0,1,2\}$ follows a binomial distribution), and allows for potential violations of this assumption.

\end{document}